\documentclass[11pt]{article}

\usepackage{fullpage}
\usepackage{amsmath,amsfonts,amsthm,mathrsfs,hyperref,graphicx}
\usepackage{bbm}
\usepackage{times}
\usepackage{amssymb,latexsym}
\usepackage{float}
\usepackage[normalem]{ulem}

\usepackage{umoline}
\usepackage[usenames,svgnames]{xcolor}

\hypersetup{
    colorlinks=true,       		
    linkcolor=Purple,          	
    citecolor=Purple,            
    filecolor=black,      		
    urlcolor=black,           	
}

\newtheorem{theorem}{Theorem}[section]
\newtheorem{proposition}[theorem]{Proposition}
\newtheorem*{proposition*}{Proposition}

\newtheorem{lemma}[theorem]{Lemma}

\newtheorem{corollary}[theorem]{Corollary}
\newtheorem{definition}[theorem]{Definition}

\newcommand{\Eq}[1]{Eq.~(\ref{#1})}
\newcommand{\Fig}[1]{Fig.~\ref{#1}}

\newcommand{\Ref}[1]{Ref.~\cite{#1}}

\newcommand{\cc}[1]{~\cite{#1}}

\newcommand{\beq}{\begin{eqnarray}}
\newcommand{\eeq}{\end{eqnarray}}

\newcommand{\ket}[1]{|#1\rangle}
\newcommand{\bra}[1]{\langle#1|}
\newcommand{\braket}[2]{\langle#1|#2\rangle}
\newcommand{\Qb}{\mathcal{H}}

\newcommand{\EqDef}{:=}

\newcommand{\DL}{\mathop{\rm DL}\nolimits}

\newcommand{\Id}{\ensuremath{\mathbbm{1}}}

\DeclareMathOperator{\poly}{poly}

\newcommand{\bigO}[1]{\ensuremath{O\left(#1\right)}}

\newcommand{\bOmega}[1]{\ensuremath{\Omega\left(#1\right)}}

\newcommand{\bTheta}[1]{\ensuremath{\Theta\left(#1\right)}}

\newcommand{\spg}{\gamma}
\newcommand{\gs}{\Gamma}

\newcommand{\uI}{\underline{I}}

\newcommand{\Span}{\ensuremath{\mathop{\rm{Span}}\nolimits}}

\newcommand{\Trim}{\ensuremath{\mathop{\rm{Trim}}\nolimits}}

\newcommand{\CC}{\ensuremath{\mathbb{C}}}
\newcommand{\N}{\ensuremath{\mathbbm{N}}}

\newcommand{\mH}{\mathcal{H}}

\newcommand{\norm}[1]{\|#1\|}

\newcommand{\eps}{\varepsilon}

\newcommand{\inote}[1]{}
\newcommand{\anote}[1]{}
\newcommand{\unote}[1]{}

\newcommand{\mpr}{\textsc{Merge}}
\newcommand{\low}{\textsc{Low-Space}}

\newcommand{\ip}[2]{\langle #1 , #2 \rangle}  
\newcommand{\ran}[1]{\mbox{range($#1$)}}

\newcommand{\gen}{\textsc{Generate}}

\newcommand{\Time}{\ensuremath{\textsc{M}}}
\newcommand{\close}{\angle_m}

\newcommand{\alone}{e^{\tilde{O}\big(\frac{1}{\gamma} \log ^3 d\big)}}
\newcommand{\altwo}{e^{\tilde{O}\big(\gamma^{-1/4}\log ^{3/4} (\frac{1}{\delta} )\log  d\big)}}
\newcommand{\althree}{e^{\tilde{O}\big(\frac{\log n}{\mu} \log^3 d\big)}}

\newcommand{\aldgs}{e^{\tilde{O}\big(\frac{1}{\gamma} \log ^3 d + \frac{1}{\gamma^{1/4}}\log ^{3/4} (\frac{1}{\delta})\log d\big)}}
\newcommand{\aldge}{e^{\tilde{O}\big(\frac{1}{\gamma} \log ^3 d + \frac{1}{\gamma^{1/4}}\log ^{3/4}(\frac{n}{\delta})\log d\big)}}
\newcommand{\allds}{e^{\tilde{O}\big(\frac{\log n}{\mu} \log ^3 d + (\frac{\log n}{\mu})^{1/4}\log ^{3/4} (\frac{1}{\delta})\log d\big)}}
\newcommand{\allde}{e^{\tilde{O}\big(\frac{\log n}{\mu} \log ^3 d + (\frac{\log n}{\mu})^{1/4}\log ^{3/4}(\frac{n}{\delta})\log d\big)}}
\newcommand{\ddelta}{10^{-5}}

\begin{document}

\title{Rigorous RG algorithms and area laws for \\low energy
eigenstates in~1D}
\author{Itai Arad\thanks{Centre for Quantum Technologies (CQT), National
University of Singapore, Singapore}
  \and Zeph Landau\thanks{Electrical Engineering and Computer Sciences, 
  University of California, Berkeley, CA 94720, U.S.A.}
  \and Umesh Vazirani\thanks{Electrical Engineering and Computer Sciences, 
  University of California, Berkeley, CA 94720, U.S.A.}
  \and Thomas Vidick\thanks{Department of Computing and Mathematical Sciences,
    California Institute of Technology, Pasadena, USA. email:
    \texttt{vidick@cms.caltech.edu}}}
\date{}
\maketitle

\begin{abstract}
One of the central challenges in the study of quantum many-body systems is the complexity of simulating them on a classical computer. A recent advance~\cite{ref:LVV2013-1Dalg} gave a polynomial time algorithm to compute a succinct classical description for unique ground states of gapped 1D quantum systems. Despite this progress many questions remained unsolved, including whether there exist efficient algorithms when the ground space is degenerate (and of polynomial dimension in the system size), or for the polynomially many lowest energy states, or even whether such states admit succinct classical descriptions or area laws.

In this paper we give a new algorithm, based on a rigorously justified RG type transformation, for finding low energy states for 1D Hamiltonians acting on a chain of $n$ particles. In the process we resolve some of the aforementioned open questions, including giving a polynomial time algorithm for
$\poly(n)$ degenerate ground spaces and an $n^{O(\log n)}$ algorithm
for the $\poly(n)$ lowest energy states (under a mild density condition). 
For these classes of systems the existence of a succinct classical
description and area laws were not rigorously proved before this work. The algorithms are natural and efficient, and for the case of finding unique ground states for
frustration-free Hamiltonians the running time is $\tilde{O}(nM(n))$, where $M(n)$ is the time required to multiply two $n\times n$ matrices.
\end{abstract}

\section{Introduction}
\label{sec:intro}

One of the central challenges in the study of quantum systems is
their exponential complexity~\cite{Feynman1982simulating}: the state
of a system on $n$ particles is given by a vector in an
exponentially large Hilbert space, so even giving a classical
description (of size polynomial in $n$) of the state is a 
challenge. The task is not impossible a priori, as the physically
relevant states lie in a tiny corner of the Hilbert space.  To be
useful, the classical description of these states must support the
efficient computation of expectation values of local observables.
The renormalization group formalism~\cite{ref:Wilson1975-RG}
provides an approach to this problem by
suggesting that physically relevant quantum states can be
coarse-grained at different length scales,
thereby iteratively eliminating the ``irrelevant'' degrees of freedom.
Ideally,  by only retaining
physically relevant degrees of freedom such a coarse-graining process successfully doubles the length scale while maintaining the total description size constant. 
This idea lies at the core of
Wilson's numerical renormalization group (NRG) approach that
successfully solved the Kondo problem~\cite{ref:Wilson1975-RG}. The
approach was subsequently improved by White\cc{ref:White1992-DMRG1,ref:White1993-DMRG2}, to obtain the famous Density Matrix
Renormalization Group (DMRG) algorithm\cc{ref:White1992-DMRG1,ref:White1993-DMRG2}, which is widely used for as a numerical heuristic for identifying the
ground and low energy states of 1D systems. 

Formally understanding the success of DMRG (and NRG) has been
extremely challenging, as it touches on deep questions about how
non-local correlations such as entanglement arise from 
 Hamiltonians with local interactions. A major advance in our understanding
of these questions came through the landmark result by Hastings~\cite{ref:Hastings2007-AL} 
bounding entanglement for gapped 1D systems with unique ground state. Hasting's work was followed by a sequence of results substantially strengthening the bounds 
(see e.g. the review article~\cite{ref:ECP2010-ALrev}). In addition to the succinct classical description guaranteed by these results,
 a recent advance~\cite{ref:LVV2013-1Dalg} gave a polynomial 
time algorithm to efficiently compute such a description. 
While the primary goal of this
paper is to present rigorous new results about the nature of
entanglement in low-energy states of 1D systems, along with efficient
classical algorithms for solving such systems, we believe that the techniques
we introduce also shed new light on the Renormalization Group
(RG) framework. 

\medskip

We  let $\mH= (\mathbb{C} ^d )^{\otimes n}$ denote the Hilbert space of $n$ particles of constant dimension $d$ arranged on a line.  We consider the class of local Hamiltonians $H= \sum_i H_i$ where each $H_i$ is a positive semidefinite operator of norm at most $1$ acting on the $i$-th and $(i+1)$-st particles.  The new algorithms apply to the following  classes of 1D Hamiltonians:

  \begin{enumerate}
    \item {\em Hamiltonians with a degenerate gapped ground 
      space (DG):} $H$ has smallest eigenvalue $\eps_0$ with
      associated eigenspace of dimension $r=\poly(n)$, and second
      smallest eigenvalue $\eps_1$ such that $\eps_1- \eps_0\geq
      \gamma$. 
      
    \item {\em Gapless Hamiltonians with a low density of 
      low-energy states (LD):} The dimension of the space of all
      eigenvectors of $H$ with eigenvalue in the range $[\eps_0,
      \eps_0 + \eta]$, for some constant $\eta>0$, is $r=\poly(n)$.
  \end{enumerate}
	
For both classes of Hamiltonians, our results show the existence of succinct representations in the form of matrix product states (MPS; see e.g.~\cite{schollwock2011density,bridgeman2016hand} for background material on MPS and their use in variational algorithms) for a basis of (a good approximation to) the ground space (resp. low energy subspace) of the Hamiltonian. The bond dimension of the MPS is polynomial in $r$ and $n$ and exponential in $\gamma^{-1}$ (under assumption (DG)) or $\eta^{-1}$ (under assumption (LD)).  The algorithms return these MPS representations in polynomial time in case (DG), and quasi-polynomial time in case (LD). For the special case of finding unique ground states for
frustration-free Hamiltonians the algorithm is particularly efficient, with a running time of $\tilde{O}(nM(n))$, where $M(n)$ is the time required to multiply two $n\times n$ matrices. 

Our assumptions are relatively standard in the literature on 1D local Hamiltonians. 
For an example of the first case, where the system has a spectral gap but the ground space is
degenerate with polynomially bounded degeneracy, see e.g.~\cite{de2010ground,bravyi2015gapped}, who consider a wide class of ``natural''
frustration-free local Hamiltonians in 1D for which the dimension of
the ground space scales linearly with the number of particles. It is
also interesting to consider the case of systems which display a
vanishing gap (as the number of particles increases), while still
maintaining a polynomial density of low-energy eigenstates (see for
instance~\cite{keating2015spectra}). The assumption of polynomial
density arises naturally as one considers local perturbations of
gapped Hamiltonians: while conditions under which the existence of a
spectral gap remains stable are known~\cite{bravyi2010topological},
it is expected that as the perturbation reaches a certain constant
critical strength the gap will slowly close; in this scenario it is
reasonable (though unproven) to expect that low-lying eigenstates should remain
amenable to analysis.

Our results should be understood in the context of a substantial
body of prior work studying ground state entanglement in 1D systems.
The techniques employed in this domain typically break down for low energy and degenerate ground states, and few results were known for these questions:
Chubb and Flammia~\cite{ref:CF2015-1Dalg}
extended the approach from~\cite{ref:LVV2013-1Dalg} and subsequent improvements by Huang~\cite{ref:Haung2014-1Dalg} to establish an
efficient algorithm (and area law) for gapped Hamiltonians with a
constant degeneracy in the ground space. Masanes\cc{masanes2009area}
proves an area law with logarithmic correction under a strong
assumption on the density of states, together with an additional assumption on the
exponential decay of correlations in the ground state. 

\medskip

Our algorithm provides a novel perspective on the well known Renormalization Group (RG) formalism within condensed matter physics~\cite{ref:Wilson1975-RG}. 
Our approach is based on the idea that if our goal is to
approximate a subspace $T$ (of low energy states, say) on $n$ qubits, 
the algorithm can make progress by locally maintaining a small dimensional subspace $S \subset \mH_A$ on a set $A$ consisting of $k$ particles, with the property that $T$ is close to $S \otimes \mH_B$, where
$B$ denotes the remaining $n-k$ particles. A major challenge here
is measuring the quality of this partial solution. This is accomplished by a suitable generalization of the
definition of a viable set introduced in~\cite{ref:LVV2013-1Dalg} to the
setting of a target subspace $T$, and is one of the conceptual
contributions of this paper (Section \ref{s:vs}). A viable set has two relevant parameters, its dimension $s$ and approximation quality $\delta$ (called the viability parameter).  
We introduce a number of procedures for manipulating viable sets (see Section \ref{ss:p}). A central procedure is {\it random projection}. This procedure drastically 
cuts down the dimension of a viable set, at the expense of degrading its viability $\delta$. Our analysis shows that to a first order, the procedure of random projection achieves a trade-off between sampled dimension and approximation quality that is such that the ratio of the sampled dimension and the overlap ($1- \delta$) is invariant (see Lemma \ref{l:rs}). A second procedure, {\it error reduction}, improves the quality of 
the viable set at the expense of increasing its dimension.  This procedure is based on the construction of a suitable 
class of approximate ground state projections (AGSPs)~\cite{ref:ALV2012-AL,ref:AKLV2013-AL} --- spectral AGSPs --- and improves the dimension-quality trade-off, at the cost of increasing the complexity of the underlying MPS representations. Setting this last cost aside, the two procedures
can  be combined  to achieve what we call {\it viable set amplification}: a reduction in the dimension of a viable set, while maintaining its viability parameter unchanged (Section \ref{ss:vsa}). Viable set amplification is key to both the area law proofs and the efficient algorithms given in this paper.

In addition to its dimension as a vector space, another important measure of the complexity of a viable set is the maximum bond dimension of MPS representations for its constituent vectors --- this may be thought 
of as a proxy for the space required to actually write out a basis for the viable set. A final procedure of {\it bond trimming} helps us keep this complexity in check (Section \ref{ss:bt}). Bond trimming provides an 
efficient procedure to replace a viable set with another one of the same dimension and similar viability parameter, but composed of vectors with smaller bond dimension, 
{\em provided} that the target subspace $T$ has a spanning set of vectors with small bond dimension --- a fact that will follow from our area laws. 

The basic building block for the algorithms in this paper combines the above procedures into a process called \mpr. \mpr\ starts with viable sets defined on adjacent sets of particles, and combines them 
into a single viable set by first taking their tensor product. This has the effect of squaring the dimension and slightly degrading the quality of the viable set. Applying  viable set 
amplification restores both dimension and quality (for suitably chosen parameters). Thus \mpr\  can be used as a building block, starting with viable sets defined on individual sites and iteratively merging results along a binary tree. Since there are only $O(\log (n))$ iterations, and the bond dimension may grow exponentially with the number of iterations, 
this only yields an $n^{O(\log n)}$ algorithm. To achieve a polynomial time algorithm, each iteration of \mpr\  is modified into a procedure \mpr' which incorporates a step of bond trimming; we refer to Section~\ref{sec:merge-process} for further discussion.

A tensor network picture of \mpr\ is provided in the
figure below.\footnote{We are grateful to Christopher T.
Chubb for originally suggesting these pictures to us.}  Beginning with inputs representing subspaces of
$\ell$ qubits shown on the left, the \mpr\  process (shown on the
right) outputs a representation of a small subspace on $2\ell$
qubits. The result is a partial isometry that is reminiscent
of a MERA~\cite{ref:Vidal2008-MERA,ref:Vidal2009-MERA}, a more complex tensor network than MPS which can in some cases arise as part of a renormalization procedure~\cite{evenbly2015tensor}.
Completing the \mpr\ process into the final algorithm, however, requires an additional step of trimming which complicates the tree-like diagram shown in the figure and results in a more complex tensor network that has no direct analogue in the literature. 
We also note that whereas RG procedures can typically be
realized as a tensor network on a binary tree (where each node
represents the partial isometry associated with selecting only a
small portion of the previous space), the use of the AGSP in our
construction allows for selection of the small subspace that can be outside the
tensor product of the previous two spaces (in this respect it may be interesting to contrast the advantage gained from AGSPs to the use of disentanglers in MERA).  

\begin{center} 
  \includegraphics[scale=.50]{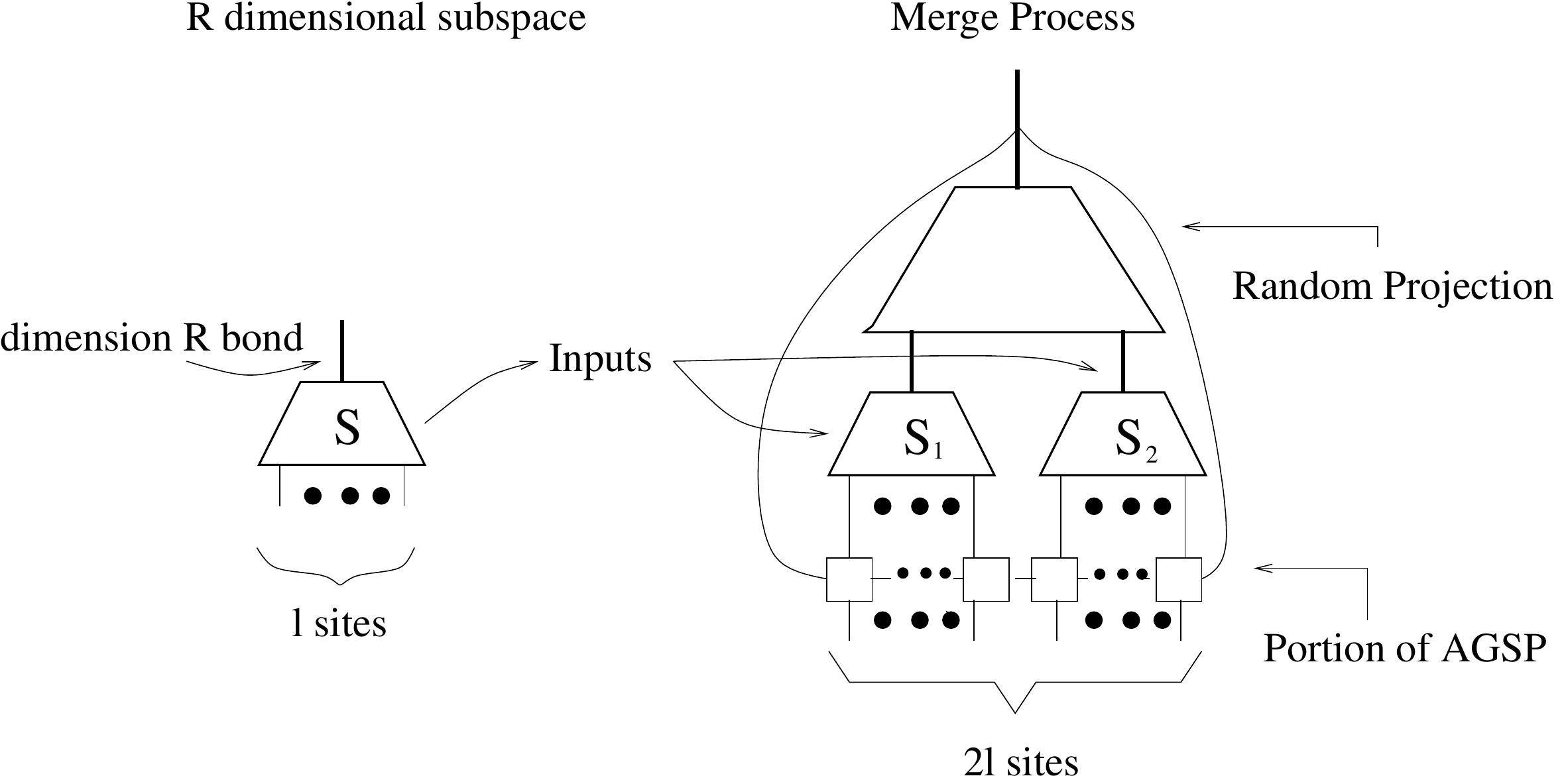} 
\end{center}


A major challenge in making the above sketch effective is the construction of appropriate AGSPs. Our new \emph{spectral AGSPs} simultaneously combine the desirable properties that had been achieved 
previously in different AGSPs. In particular, they are efficiently computable, have tightly controlled bond dimension (the parameter $D$) at two pre-specified cuts, and have bond dimension bounded by a polynomial in $n$ at every other cut. Achieving this requires a substantial amount of technical work,
building upon the Chebyshev construction of~\cite{ref:AKLV2013-AL}, ideas about soft truncation of Hamiltonians (providing efficient means of achieving similar effects to the hard truncation studied in e.g.~\cite{AradKZ14energy}), a series expansion of
$e^{-\beta H}$ known as the \emph{cluster
expansion}~\cite{hastings2006solving,KlieschGKRE14cluster}, as well as a recent
nontrivial efficient encoding of the resulting operator due
to~\cite{MolnarSVC15gibbsmpo}. The constructions of spectral AGSP appear in Section~\ref{sec:arealaws} (non-efficient constructions) and Section~\ref{sec:efficient} (efficient constructions).


\medskip

Our new algorithms could potentially be made very
efficient. The main bottlenecks are the complexity of the AGSP and
the MPS bond dimension that must be maintained. In the case of a
frustration-free Hamiltonian with unique ground state we obtain a
running time of $O(2^{O(1/\gamma^2)}n^{1+o(1)}\Time(n))$, where
$\Time(n)$ is matrix multiplication time. This has an exponentially
better scaling in terms of the spectral gap $\spg$ (due to avoidance
of the $\eps$-net argument) and saves a factor of $n/\log n$ (due to
the logarithmic, instead of linear, number of iterations) as
compared to an algorithm for the same problem considered
in~\cite{ref:Huang-copy}. We speculate that it might further be
possible to limit the bond dimension of all MPS considered to
$n^{o(1)}$ (instead of $n^{1+o(1)}$ currently), which, if true,
would imply a nearly-linear time $O(n^{1 + o(1)})$ algorithm.

Subsequently to the completion of our work, a heuristic variant of the algorithm described in this paper has been implemented numerically~\cite{roberts2017rigorous}. Although this initial implementation typically suffers from a $\sim 5-10\times$ slowdown compared to the well-established DMRG, it provides encouragingly accurate results, matching those of DMRG in ``easy'' cases, but also sometimes outperforming DMRG, e.g. in cases where the ground space degeneracy is high (linear in system size) or for some critical systems.

\paragraph*{Organization} The remainder of the paper is organized as follows. In Section~\ref{s:vs} we start by introducing viable sets, and provide a comprehensive set of procedures to work them; these procedures form the building blocks of our area laws and algorithms. With these procedures in place, in Section~\ref{sec:overview} we  provide an overview of our proof technique; this section may be the best place to start reading the paper for a reader new to our results. The following three sections are devoted to a formal fleshing out of our results. In Section~\ref{sec:arealaws} we prove our area laws by showing the existence of good AGSP constructions. In Section~\ref{sec:efficient} we provide efficient analogues of these AGSP constructions, which are employed in Section~\ref{sec:algorithms} to derive our efficient algorithms. We conclude in Section~\ref{sec:discussion} with a discussion of our results and possible improvements.

\section{Viable sets}
\label{sec:viable}  \label{s:vs}

Our approach starts with the idea that the challenge of finding a solution --- a low-energy state --- within a Hilbert Space $\Qb$ of exponential size can be approached by starting with ``partial solutions'' on small subsystems, and gradually combining those into  ``solutions'' defined on larger and larger subsystems.   To implement this approach we need a formal notion of partial solution, as well as techniques for working with them.  This is done in the next few subsections where we introduce {\it viable sets} to capture ``partial solutions'', and describe  procedures to efficiently work with such viable sets.

\subsection{Definition and basic properties}

Given a subset $A$ of
particles, we may decompose the full Hilbert space of the system as a
tensor product $\mH=\mH_A\otimes\mH_B$, where $\mH_A$ is the Hilbert space
associated with particles in $A$ and $\mH_B$ is the Hilbert space associated with the remaining particles in the system. Our ultimate goal is to compute (an approximation to) some subspace $T\subset \mH$. Towards this we wish to 
measure partial progress made while processing only particles in the subset $A$. This can be expressed through the sub-goal 
of finding a subspace $S\subset \mH_A$ with the guarantee that 
$S_{ext} \EqDef S\otimes \mH_B$ contains $T$.  Since we need to allow the possibility for 
approximation errors, we are led to the following definition:

\begin{definition}[{\bf Viable Set}] \label{d:vs}
  Given $0\leq \delta \leq 1 $ and a subspace $T\subseteq \mH=\mH_A\otimes \mH_B$, a
  subspace $S\subseteq \mH_A$ is \emph{$\delta$-viable} for $T$ if
  \begin{equation}\label{vsdef} 
    P_T P_{S_{ext}} P_T \geq (1- \delta) P_T, 
  \end{equation}
  where $S_{ext}\EqDef S\otimes\mH_B$ and $P_T$ (resp. $P_{S_{ext}}$) is the orthogonal projection onto $T$ (resp. $S_{ext}$). We refer to $\delta$ as the \emph{viability} of the set, and $\mu=1-\delta$ as its \emph{overlap}. 
\end{definition}

This definition  
captures the notion that a reasonable approximation of $T$ can be found within the subspace $S \otimes \mH_B$. It generalizes the definition of a viable set 
from~\cite{ref:LVV2013-1Dalg}, which was specialized to the case where $T$ is a
one-dimensional subspace containing a unique ground state. In~\cite{ref:CF2015-1Dalg,ref:Haung2014-1Dalg} the definition was extended to handle degenerate ground spaces by explicitly requiring that the viable set support orthogonal vectors that are good approximations to orthogonal ground states. Here we avoid making any direct reference to a basis, or families of orthogonal vectors, and instead work directly with subspaces. 
 
While the notion of viable set is quite intuitive for small $\delta$, our arguments  also involve viable sets with parameter $\delta$ close to $1$ (alternatively, $\mu = 1-\delta$ close to $0$, where $\mu$ is a parameter we will refer to as the \emph{overlap} of the viable set), a regime where there is less intuition. A helpful interpretation of the definition is that it formalizes the fact that for a viable set $S$, the image of the unit ball of $S_{ext}$ when projected to $T$ contains the ball of radius $(1-\delta)$. 

\begin{lemma} \label{l:2} 
  If $S$ is $\delta$-viable for $T$ for some $0\leq \delta <1$ then for every $\ket{t} \in T$ of
  unit norm, there exists an $\ket{s}\in S_{ext}$ such that $P_T\ket{s}=\ket{t}$ and $\norm{\ket{s}}\le  \frac{1}{1-\delta}$.
\end{lemma}

The proof of Lemma \ref{l:2} follows directly from the  following general operator facts:
 
\begin{lemma}
\begin{enumerate}
\item If $X$ and $Y$ are positive operators  and $X \geq Y$ then range ($Y$)$ \subset$ range ($X$).
\item  If $PQP \geq cP$ for projections $P,Q$ and $c>0$ then for every $v \in \ran{P}$ of unit norm, there exists $w \in$ range($Q$), $||w||=1$ such that $Pw=c_v v$ for some constant $c_v$ with $|c_v|\geq c$.
\end{enumerate}
\end{lemma}
\begin{proof}  For 1., suppose $y \in \ran{Y}$ and let $y= x + x_{\perp}$, $x\in \ran{X}$, $x_{\perp} \perp \ran{X}$ be the orthogonal decomposition.  Since $\ip{Xx_{\perp}}{x_{\perp}}=0$ it follows that $\ip{Yx_{\perp}}{x_{\perp}}=0$ and thus $x_{\perp} \perp \ran{Y}$ as well and hence $x_{\perp}=0$ and $y =x \in \ran{X}$.

For 2., it follows from 1.  that if  $PQP \geq c P$ then for any $v \in \ran{P}$ there exists an $r \in \ran{P}$ such that $PQPr =PQr= v$.   So then $\ip{PQPr}{r} \geq c \ip{Pr}{ r}= c||r||^2$.  But $\ip{PQPr}{r}= \ip{v}{r}\leq ||r||||v||$.  Putting these two inequalities together along with the assumption that $||v||=1$ yields $ ||r|| \leq 1/c$.  
\end{proof}

We introduce a notion of proximity between subspaces:

\begin{definition}[\bf Closeness]\label{d:close}
For $0\leq \delta \leq 1$, a subspace $T$ is {\em $\delta$-close} to a subspace $T'$ if 
\[ P_{T'} P_{T} P_{T'} \geq (1-\delta) P_{T'}, \]
where $P_T$ and $P_{T'}$ are the orthogonal projections on $T$ and $T'$ respectively. 
We say that $T$ and $T'$ are {\em mutually $\delta$-close} if each is $\delta$-close to the other, and denote by $\close(T,T')$ the smallest $\delta$ such that $T,T'$ are mutually $\delta$-close. 
\end{definition}


Closeness of subspaces is approximately transitive in the following sense:

\begin{lemma}[{\bf Robustness}] \label{l:close}
If $T$ is $\delta$-close to $T'$ and $T'$ is $\delta'$-close to $T''$ then $T$ is $2(\delta + \delta')$-close to $T''$.  Consequently if  $S$ is $\delta$-viable for $T$ and $T$ is $\delta'$-close to $T'$ then $S$ is $2(\delta + \delta')$-viable for $T'$.
\end{lemma}

 \begin{proof}     Notice that $P_A P_B P_A \geq (1-\delta)P_A$ is equivalent to the statement that $\norm{P_B \ket{a}}^2\geq (1- \delta)$ for all $\ket{a}\in A$ with $\norm{\ket{a}}=1$.  It follows for  $\ket{t''} \in T'$ of unit norm, $\ket{t'}=P_{T'}\ket{t''}$ has the property that $\norm{\ket{t'}}^2 \geq (1 -\delta')$ and thus $\norm{\ket{t'}- \ket{t}} \leq \sqrt{\delta'}$.  Similarly $\ket{t}=P_{T} \ket{t'}$ has the property that $\norm{\ket{t}}^2 \geq \norm{\ket{t'}}^2 (1- \delta)$ and thus $\norm{\ket{t'} -\ket{t}} \leq \norm{\ket{t'}} \sqrt{\delta}$.  By the triangle inequality, $\norm{\ket{t''} - \ket{t}} \leq \sqrt{\delta'} + \sqrt{\delta}$ and since $\ket{t} \in T$, this implies that the distance between $\ket{t''}$ and $T$ is at most $\sqrt{\nu} + \sqrt{\delta}$, i.e. $\norm{P_{S_{ext}} \ket{t}}^2 \geq 1 - (\sqrt{\nu} + \sqrt{\delta})^2 \geq 1- 2(\nu + \delta)$.  As mentioned, this last statement is equivalent to $T$ being a $2(\nu + \delta)$ close to $T''$.
  \end{proof}

\subsection{Procedures}  \label{ss:p}

We introduce a set of procedures that can be performed on viable sets.  These procedures will allow us to build viable sets on larger and larger subsystems, while keeping the complexity and size of the viable sets small.  They will serve as the core operations for  both our area law proofs and our algorithms.  

\subsubsection{Tensoring}

The next lemma summarizes the effect of tensoring two viable sets supported on disjoint subsystems. 

\begin{lemma}[{\bf Tensoring}]\label{lem:merging} \label{l:t}
  Suppose $S_1$, $S_2$ are $\delta_1$-viable and $\delta_2$-viable
   for $T$ respectively, defined on disjoint subsystems. Then the set $S\EqDef S_1\otimes
  S_2$  is   $(\delta_1+\delta_2)$-viable for $T$.
  \end{lemma}

  \begin{proof}
  Since $S_1, S_2$ are defined on disjoint subsystems, it
  follows that $P_{S^{(ext)}} = P_{S_1^{(ext)}} 
  P_{S_2^{(ext)}}$, and so
  \begin{align*}
    P_T P_{S^{(ext)}} P_T = P_T P_{S_1^{(ext)}} P_{S_2^{(ext)}} P_T
     = P_T P_{S_1^{(ext)}} P_T 
       - P_T P_{S_1^{(ext)}} \big(\Id-P_{S_2^{(ext)}}\big)P_T \,.
  \end{align*}
  The definition of a viable set implies that
  $P_T P_{S_1^{(ext)}} P_T  \ge (1-\delta_1) P_T$. In addition,
  \begin{align*}
    P_T P_{S_1^{(ext)}} \big(\Id-P_{S_2^{(ext)}}\big)P_T 
      \le P_T \big(\Id-P_{S_2^{(ext)}}\big)P_T
      \le \delta_2 P_T \,.
  \end{align*}
  Therefore, $P_T P_{S^{(ext)}} P_T \ge (1-\delta_1-\delta_2)P_T$.
  \end{proof}

\subsubsection{Random sampling}
	
	The following lemma establishes how viability of a set is affected when sampling a random subset. 
	
	
	\begin{lemma}[{\bf Random sampling}]\label{lem:viable-sample} \label{l:rs}
	Let  $T\subseteq \mH = \mH_L \otimes \mH_M \otimes \mH_R$ be an  $r$-dimensional subspace, and  $W$ a  $q$-dimensional subspace of $\mH_M$ that is $\delta$-viable for $T$.  
	Then a random $s$-dimensional subspace $W'$ of $W$ is $(1-\delta')$-viable for $T$ with probability
	$1- \eta$, where
	$$\delta' = \frac{(1-\delta)}{8}\frac{s}{q}\qquad\text{and}\qquad \eta =\Big(1+4\sqrt{\frac{q}{(1-\delta) s}}\Big)^r q e^{-s/16} .$$
	\end{lemma}
	
	
	\begin{proof}
	Let $\ket{v}$ in $T$ such that $\|\ket{v}\|=1$, and $\ket{w} = P_{W_{ext}} \ket{v}\in W_{ext}$. Using that $W$ is $\delta$-viable for ${T}$ it follows that $\norm{\ket{w}}^2 \geq 1-\delta$. Since $W'_{ext} \subseteq W_{ext}$, $P_{W'_{ext}} \ket{v}=P_{W'_{ext}} \ket{w}$.  
By a standard concentration argument based on the Johnson-Lindenstrauss lemma (see e.g.~\cite[Theorem 2.1]{ref:DG2003-elementaryJL}) it holds that  $\norm{P_{W'_{ext}}\ket{v}}^2 \geq (1-\delta) \frac{s}{2q}$  with probability at least $1- q e^{-s/16}$. Let $\nu= \sqrt{(1-\delta)s/8q}$. By a volume argument (see e.g.~\cite[Lemma~5.2]{ref:Ver2010-randMat}), there exists a subset $S$ of the Euclidean unit ball of $T$ such that $|S|\leq(1+ 2/\nu)^r$ and for any unit $\ket{t}\in T$, there is an $\ket{v}\in S$ such that $\|\ket{s}-\ket{t}\|\leq \nu$. Applying the preceding argument to each $\ket{v}$ in the net, by the union bound the choice of $\eta$ made in the theorem is with probability at least $1-\eta$,  $\norm{P_{W'_{ext}}\ket{v}}^2 \geq (1-\delta)s/(2q)$ for all $\ket{v}$ in the net; hence $\norm{P_{W'_{ext}}\ket{v}}^2 \geq (1-\delta) s/(8q)$ for all $\ket{v}$ in the unit ball of $T$.
\end{proof}
	
\subsubsection{Error reduction using Approximate Ground State Projections}
\label{sec:agsp-def}

We address the question of how to improve the viability parameter $\delta$ for a given viable set.    In previous work this question was addressed for the case of the target space $T$ being one dimensional by introducing the key tool of Approximate Ground State Projections (AGSPs)~\cite{ref:ALV2012-AL,ref:AKLV2013-AL}. AGSPs have been used in the context of proofs of the 1D area law for Hamiltonians with a unique ground state as well as  in algorithms for finding the ground state of a
gapped 1D system~\cite{ref:LVV2013-1Dalg}.   

Whereas in previous works an AGSP was primarily constructed to
approximate the projector on a \emph{unique} ground state, here our
main focus is on the case of a degenerate ground space and low-energy states. 
We therefore introduce a
more general definition of an AGSP as a local
operator that increases the norm of eigenvectors in the low part of the spectrum of
$H$, while decreasing the norm of eigenvectors in the high energy part of the
spectrum. We refer to this object as a
\emph{spectral AGSP}.

\begin{definition}[{\bf Spectral AGSP}]\label{def:spectral-AGSP}
  Given $\mH = \mH_L\otimes \mH_M\otimes \mH_R$, $H$ a Hamiltonian on $\mH$ and $\eta_0<\eta_1$, a positive semidefinite operator
  $K$ on $\mH$ is a $(D,\Delta)$-spectral AGSP for $(H,\eta_0,\eta_1)$ if the
  following conditions hold:
  \begin{itemize}  
  \item $K$ has a decomposition $K= \sum_{i=1}^{D^2} L_i \otimes A_i \otimes R_i$,
		\item $H$ and $K$ have the same eigenvectors,
    \item Eigenvalues of $H$ smaller than $\eta_0$ correspond to eigenvalues of $K$ that are larger than or equal to $1$,
    \item Eigenvalues of $H$ larger than $\eta_1$ correspond to eigenvalues of $K$ that are smaller than $\sqrt{\Delta}$. 
  \end{itemize}
	\end{definition}

The advantage of an AGSP, compared to an exact projection operator, lies in the fact that 
one can often construct a much more \emph{local} operator, i.e., an operator
with a much smaller Schmidt rank compared to the exact projector. 
The existence of an AGSP of small Schmidt rank which greatly shrinks the high energy part
of the spectrum can be viewed as a strong characterization of
the locality properties of the low-energy space.  A favorable scaling
between these two competing aspects (in the case of unique ground states) was the key feature in recent
proofs of the 1D area law~\cite{ref:ALV2012-AL,ref:AKLV2013-AL} via the
bootstrapping lemma. 
The following lemma establishes a lower bound on the quantitative improvement in viability that a spectral AGSP can achieve on a viable set.


\begin{lemma}[{\bf Error reduction --- Spectral AGSP}]\label{lem:error-spectral} \label{l:er}
Let $\mH=\mH_L\otimes \mH_M\otimes \mH_R$, $H$ a Hamiltonian on $\mH$,  $\eta_0<\eta_1$, and $K= \sum_{i=1}^{D^2} L_i \otimes A_i \otimes R_i$ a $(D,\Delta)$-spectral AGSP for $(H,\eta_0,\eta_1)$ where $H$ has ground state energy $\eps_0$ and has no eigenvalues in the interval $(\eta_0, \eta_1)$. Let $S\subseteq \mH_M$ be a $\delta$-viable set for $T = H_{[\eps_0,\eta_0]}$ of dimension $s$.
 Then the space $V= \Span \{ A_{i} S: 1\leq i \leq D^2 \}$ has dimension at most $D^2s$ and is $\delta'$-viable for $T$ with 
$$\delta'= \frac{\Delta}{(1-\delta)^2}.$$ 
\end{lemma}

 \begin{proof}
The bound on the dimension of $V$ is straightforward.  To show $V$ is $\delta'$-viable for $T$, begin with an arbitrary unit norm vector $\ket{v}\in T$.  Set $\ket{v'} = \frac{1}{\norm{K^{-1}\ket{v}}} K^{-1} \ket{v}$, where $K^{-1}$ is the pseudo-inverse. Then $\ket{v'}$ is also an element of $T$.  Since $S$ is $\delta$-viable for $T$, applying Lemma~\ref{l:2} there exists an $\ket{u}\in \mH_L\otimes S \otimes \mH_R$ whose projection onto $T$ is, up to scaling, precisely $\ket{v'}$; thus $\ket{u} = \alpha \ket{v'} + \sqrt{1-\alpha^2}\ket{v^\perp}$ for some $\alpha \geq 1-\delta$ and unit $\ket{v^\perp}$ that is orthogonal to $T$. In particular $\ket{v^\perp}$ is supported on the span of all eigenvectors of $H$ with eigenvalue outside of $[\eps_0,\eta_1)= [\eps_0, \eta_0) \cup [\eta_0, \eta_1)$ and thus by the property of $K$, $\|K\ket{v^\perp}\|^2 \leq \Delta$. 

Applying $K$ to $\ket{u}$ yields $K\ket{u}= \alpha'\ket{v} + K\ket{v^\perp}$ with $\alpha'=\alpha \frac{1}{\norm{K^{-1}\ket{v}}} \geq \alpha$ (since $\ket{v}$ is supported on eigenvectors of $K$ with corresponding eigenvalue at least $1$). Thus

\begin{align*}
\Big|\Big\langle \frac{ K u}{\| K \ket{u}\|}\Big|v\Big\rangle\Big|^2 &\geq \frac{\alpha'^2 }{\alpha'^2+ (1-\alpha'^2)\Delta}\\
& \geq 1-\frac{1}{ (1-\delta)^2}\Delta.
\end{align*}
\end{proof}

\subsubsection{Complexity reduction using trimming} \label{ss:bt}

For a viable set to be efficiently represented it must not only have small dimension but also a basis of states that can be efficiently described, say by polynomial-bond matrix product states.  A natural question is, assuming that the target subspace $T$ has a basis of vectors of small bond dimension, whether it is possible to efficiently ``trim'' any sufficiently good viable set for $T$ into another almost-as-good viable set specified by vectors with comparably small bond dimension.  

To achieve this goal we introduce a modified trimming procedure to that of \cite{ref:LVV2013-1Dalg}.  There the trimming procedure is based on the observation that given a
good approximation to a target vector $\ket{v}$ of low bond
dimension, trimming the approximating vector by dropping Schmidt
vectors associated with the smallest Schmidt coefficients at each
cut yields an almost-as-good approximation to $\ket{v}$ with lower
bond dimension. In the present scenario the approximating vector is
not known: instead we are given a basis for a subspace that contains
the approximating vector. A natural idea would be to trim the MPS
representations for the basis vectors in a way that guarantees that
$\ket{v}$ is still closely approximated by some vector in the span
of the resulting set. However, it is not clear if independently trimming each
of the basis vectors, as done in~\cite{ref:LVV2013-1Dalg}, works --
indeed, the basis vectors themselves could a priori have a very flat
distribution of Schmidt coefficients, so that trimming could induce large changes.

We provide a modified procedure which starts with a basis
for the viable set and trims the basis vectors collectively at each
cut, from the leftmost to the rightmost, as follows (informally): for each
cut, project each element of the basis onto the span of the left
Schmidt vectors of \emph{any} basis element that is associated with
a large Schmidt coefficient.

\begin{definition}[{\bf Trimming}]\label{def:trimming}
  Let $S\subseteq \mH_A$ be a $\delta$-viable set for $T\subseteq
  \mH_A\otimes\mH_B$ specified by an orthonormal basis
  $\{\ket{u_i},\,i=1,\ldots s\}$. Suppose $\mH_A =
  \mH_A^1\otimes\cdots\otimes\mH_A^\ell$ for some $\ell\geq 2$. Let $\ket{\psi}
  = \sum_i \ket{u_i}\ket{i} \in \mH_A \otimes \CC^s$. For $j$ from $1$ to $(\ell-1)$ define $P_{\geq
  \xi}^j$ inductively as the projection on the subspace of $P_{\geq \xi}^{j-1} \otimes \Id_{\mH_A^j}$ spanned by the left Schmidt
  vectors of $P_{\geq
  \xi}^{j-1}\otimes \Id_{\mH_A^j \cdots\mH_A^{\ell}} \otimes \Id_{\CC^s} \ket{\psi}$ across the cut $(j:j+1)$ with associated
  Schmidt coefficient at least $\xi$.\footnote{Note that we do not re-normalize vectors.} Then the \emph{$\xi$-trimmed
  set} is
  \begin{equation}\label{eq:def-viable-trim}
    \Trim_\xi(S) \EqDef \Span\Big\{\big(( P^1_{\geq\xi}\otimes \Id_{\mH_A^2\otimes\cdots\otimes \mH_A^k})\cdots (P^{\ell-1}_{\geq \xi}\otimes\Id_{\mH_A^\ell})\big)\ket{u_i},\,i=1,\ldots,s\Big\}.
  \end{equation}
\end{definition}

With this notion of trimming,  we show that if a set $S$ is a good viable set for a set $T$ whose elements are guaranteed to have low bond dimension then the result of trimming the set $S$ does not degrade the quality of the viable set too much.  

\begin{lemma}[{\bf Trimming}]\label{lem:trimming}
Let $S\subseteq \mH_A$ be a $\delta$-viable set of dimension $s$ for $T\subseteq \mH_A\otimes\mH_B$. Suppose $\mH_A = \mH_A^1\otimes\cdots\otimes\mH_A^{\ell}$ for some ${\ell\geq 2}$. Let $b$ be an upper bound on the Schmidt rank of any vector in $T$ across any cut $(j:j+1)$ for $j=1,\ldots,{\ell}-1$. Then the $\xi$-trimmed set $\Trim_\xi(S)$ is a $\delta'$-viable set for $T$ for $\delta'\leq \delta + \sqrt{{\ell}bs}\xi$. 

Furthermore, a spanning set for $\Trim_\xi(S)$ containing at most $s$ vectors of Schmidt rank at most $s\xi^{-2}$ across any cut can be computed in time $O(\ell \Time(dsq))$, where $q$ is an upper bound on the bond dimension of MPS representations for a basis of $S$ and $\Time(\cdot)$ denotes matrix multiplication time.
\end{lemma}

\begin{proof}
Let $\{\ket{u_i},\, i =1,\ldots,s\}$ denote an orthonormal basis for $S$, and $\ket{v}\in T$ a unit vector. Let $\ket{u}=\sum_i \mu_i\ket{a_i}\ket{b_i} \in \mH_A \otimes \mH_B$ be a unit vector such that $|\bra{u} v\rangle|^2 \geq 1-\delta$. For $j=0,\ldots,{\ell}$, let 
$$\ket{u'_j} = (P^1_{\geq\xi}\otimes \Id_{\mH_A^{2}\cdots \mH_A^{\ell}}\otimes\Id_{\mH_B}) \cdots (P^j_{\geq\xi}\otimes \Id_{\mH_A^{j+1}\cdots \mH_A^{\ell}}\otimes\Id_{\mH_B}) \ket{u},$$
and for $i\in\{1,\ldots,s\}$,
$$\ket{a_i^j} = (P^1_{\geq\xi}\otimes \Id_{\mH_A^{2}\cdots \mH_A^{\ell}}\otimes\Id_{\mH_B}) \cdots (P^j_{\geq\xi}\otimes \Id_{\mH_A^{j+1}\cdots \mH_A^{\ell}}\otimes\Id_{\mH_B}) \ket{a_i}.$$
By definition of the $P^j_{\geq\xi}$ (Definition~\ref{def:trimming}), the Schmidt coefficients of the vector 
$$(P^1_{\geq\xi}\cdots P_{\geq
  \xi}^{j-1} (\Id-P^j_{\geq\xi})\otimes \Id)\ket{\psi},$$ where $\ket{\psi} = \sum \ket{a_i}\ket{i}$, across the cut $(j,j+1)$ are all at most $\xi$. Since acting with a local projection (here, $\ket{i}\bra{i}$ on $\mH_B$) cannot increase the largest Schmidt coefficient, the same holds of the vector $((\Id-P^j_{\geq\xi})\otimes \Id)\ket{a_i^{j-1}}$. Based on these observations we may upper bound, for any $i,j$, and unit $\ket{c} \in \mH_A^1 \otimes \cdots \otimes \mH_A^j$ and $\ket{d} \in \mH_A^{j+1} \otimes \cdots \otimes \mH_A^\ell\otimes \mH_B$,
\begin{align*}
 \big|\bra{a_i^{j-1}}\bra{v_i} ((\Id-P^j_{\geq\xi})\otimes \Id_{\mH_A^{j+1}\otimes\cdots\otimes \mH_A^{\ell}}\otimes \Id_{\mH_B})\ket{c}\ket{d}\big|
&\leq \xi,
\end{align*}
where the inequality follows since we are taking the inner product of a vector with largest Schmidt coefficient at most $\xi$ with another vector of Schmidt rank $1$. Using the promised bound on the Schmidt rank of $\ket{v}$ we deduce 
\begin{align*}
\big|(\bra{u'_j}-\bra{u'_{j-1}})|v\rangle \big| &= \big|\bra{a_i^{j-1}}\bra{b_i} ((\Id-P^j_{\geq\xi})\otimes \Id |v\rangle \big|\\
&\leq \xi \sqrt{bs}\big\|(\Id-P^j_{\geq\xi})\otimes \Id\ket{v}\big\|.
\end{align*}
Using a telescopic sum, and orthogonality of the projections $(\Id-P^j_{\geq\xi})\otimes \Id$ for different values of $j$, we get
\begin{align*}
\big|(\bra{u'_1}-\bra{u'_{\ell}})\ket{v}\big|^2 &\leq \xi^2 bs \Big(\sum_{j=1}^\ell \big\|(\Id-P^j_{\geq\xi})\otimes \Id\ket{v}\big\|\Big)^2 \\
&\leq \xi^2 \ell bs,
\end{align*}
 and the claimed bound on $\delta'$ follows. 

For the ``furthermore'' part, note that $\ket{\psi}$ has at most $s/\xi^2$ Schmidt coefficients larger than $\xi$ across any cut $(j:j+1)$. Thus  each $P^{j}_{\geq \xi}$ has rank at most $s/\xi^2$, so that its application reduces the Schmidt rank across the cut $(j:j+1)$ to at most $s/\xi^2$, while not increasing it to a larger value at any of the previously considered cuts. The left Schmidt vectors of
$$ ( P^1_{\geq\xi}\otimes \Id_{\mH_A^2\otimes\cdots\otimes \mH_A^{\ell}})\cdots (P^{{\ell}-1}_{\geq \xi}\otimes\Id_{\mH_A^{\ell}}) \big)\ket{\psi}$$
across the cut specified by the division $\mH = \mH_A\otimes \mH_B$ form a spanning set for $\Trim_\xi(S)$. 

In order to compute canonical MPS representations for a basis of $\Trim_\xi(S)$ we first create an MPS representation for $\ket{\psi}$ and reduce it to canonical
form (we refer to e.g. the survey~\cite{verstraete2008matrix} for a discussion of basic operations on MPS and their computational efficiency). This costs $O(\ell \Time(dsq))$ operations, where $\Time(\cdot)$ is matrix multiplication time, and $\Time(dsq)$ is the time required to perform required basic operations on tensors of bond dimension $O(dsq)$, such as singular value decompositions. Proceeding from the cut $({\ell}-1,{\ell})$ to the $(1,2)$ cut from right to left, we then set the coefficients of the diagonal tensor matrices $\Lambda_j$ from the MPS representation  that are smaller than $\xi$ to zero. The resulting re-normalized state is automatically
given in canonical MPS form, and a spanning set for $\Trim_\xi(S)$ can be obtained by cutting the last bond.
\end{proof}

\section{Overview}
\label{sec:overview}

In this section we provide an outline of how the procedures introduced in the two previous sections can be put together to yield area laws and efficient algorithms.  Our results hinge on our ability to construct AGSPs with good trade-offs between $D$ and $\Delta$. Our goal in this section is to provide a high level picture of how the pieces fit together. For this we assume a very simple, approximate picture of an AGSP. The rigorous results are more intricate, and will be described in the remaining sections of the paper. 

Let $H$ be a Hamiltonian with ground state energy $\eps_0$ and no eigenvalues in the interval $(\eta_0, \eta_1)$. We assume that $H$ comes with an associated spectral AGSP $K$ (Definition \ref{def:spectral-AGSP}) that satisfies the conditions of Lemma~\ref{l:er}. We further assume that the parameters $(D,\Delta$) associated with  $K$ satisfy a sufficiently good trade-off between $D$ and $\Delta$.\footnote{For our purposes, a tradeoff of the form $D^{c} \Delta < \frac{1}{2}$, for a large enough constant $c$, will suffice; we refer to later sections for concrete parameters.}  Our goal is to approximate the low-energy subspace $T = H_{[\eps_0,\eta_0]}$, assumed to be of polynomially bounded dimension.

\subsection{Viable set amplification and area laws}
\label{sec:viable-amplification}

As a first step we compose the procedures of random sampling (Lemma \ref{l:rs}) and error reduction (Lemma \ref{l:er}) to obtain a procedure that improves the quality of a viable set without increasing its dimension:

\medskip

\noindent{\bf{Viable Set Amplification:}} \label{ss:vsa}\\
Given is a $\delta$-viable set $W$ of dimension $q$.
\begin{enumerate}
\item Generate a {\bf random sample}, as in Lemma \ref{l:rs}, to obtain $S\subset W$ of dimension $s$ with viability parameter $\delta' > \delta$. 
\item Apply {\bf error reduction} to $S$, as in Lemma \ref{l:er}, using the AGSP $K$, to produce a $\delta$-viable set $W'$ of dimension $q'$.
\end{enumerate}

For a $\delta$-viable set $W$, we refer to $\mu = 1-\delta$ as its overlap. Random sampling reduces the dimension of the viable set but also proportionately reduces its overlap. The second step (AGSP) increases the overlap at the cost of a comparatvely smaller increase in dimension --- a favorable trade-off due to the favorable $D-\Delta$ trade-off of the AGSP. With proper setting of parameters, the viable set amplification procedure above reduces the dimension of the viable set while leaving the overlap (and $\delta$) unchanged, as long as the viable set dimension $q> q_0$, for some $q_0$ determined by $\delta$ as well as the parameter $D$ of the AGSP (itself related to parameters of the initial Hamiltonian, including the spectral gap above the low-energy space $T$). 

Reasoning by contradiction, the argument implies the existence of a $\delta$-viable set $W_0$ for $T$ of dimension at most $q_0$.  The existence of such a $W_0$ in turn implies that any element of $T$ has a $\delta$-approximation by a vector with entanglement rank no larger than $q_0$. An area law follows easily using standard amplification arguments; we give the details in Section~\ref{sec:arealaws}.

\subsection{Merge process and algorithms}
\label{sec:merge-process}

In the argument described in the previous section the  parameters were chosen such that a $\delta$-viable set of dimension $q$ was ``amplified'' to a $\delta$-viable set of dimension $q'<q$.  With a slightly more demanding choice of parameters viable set amplification can be made to reduce both the dimension $q \rightarrow q'=\sqrt{q}$ and the viability parameter $\delta \rightarrow \delta'=\frac{\delta}{2}$.  This only requires a slightly more stringent condition on the $D-\Delta$ trade-off provided by the underlying AGSP.   

We now explain how viable set amplification can be folded within a larger procedure that we call \mpr. Assume given a decomposition
$\mH=\mH_L\otimes (\mH_A \otimes \mH_B )\otimes \mH_R$ of the $n$-particle Hilbert space. \mpr\  starts with two  viable sets $V_1\subseteq \mH_A$ and $V_2 \subseteq \mH_B$ and returns a viable set $V\subseteq \mH_A \otimes \mH_B$.  It does so in a way such that all parameters of the viable set $V$, namely the viability $\delta$, the dimension, and its description complexity, are comparable to those of the original two sets. We proceed to describe \mpr; for expository purposes we set aside considerations on the complexity of representing elements of the viable sets (these will be made formal in subsequent sections). 

\medskip

\noindent{\bf{\mpr}:}\\
Given are two $\delta$-viable sets $V_1 \subset \mH_A$ and $V_2 \subset \mH_B$ of dimension $q$. \begin{enumerate}
\item {\bf Tensor} the two sets, as in Lemma \ref{l:t}, to obtain a $2\delta$-viable set  $W= V_1 \otimes V_2$ of dimension at most~$q^2$.
\item Perform {\bf viable set amplification}  to yield a $\delta$-viable set $V \subset \mH_1 \otimes \mH_2$ of dimension at most $q$.
\end{enumerate}

Our algorithm starts with (easily generated) viable sets defined over small subsets of particles, and iterates \mpr\ in a tree-like fashion to eventually generate a single viable set defined over the entire space.  With this final viable set in hand, it is not difficult to find low-energy states within the viable set, {\it provided} we are able to describe its elements using low-complexity representations (e.g. low bond dimension matrix product states).  This will not be the case unless explicit constraints are enforced on the complexity of the operators used in the error reduction step of viable set amplification, where the complexity can blow up rapidly due to the application of the AGSP $K$.  

To maintain the desired low complexity MPS representations and complete the algorithm we make two modifications to \mpr.  The first is within the AGSP construction, where  a procedure of soft truncation (Section~\ref{sec:soft-truncation}) leads to the operators used in error reduction having matrix product operator (MPO) representations with polynomial bond dimension. Since these operators are applied a large number of times, however, the complexity of the MPS representations manipulated could still increase to super-polynomial. In order to keep that complexity under control we perform a second modification, which decomposes the viable set amplification procedure into smaller steps of viable set amplification followed by a trimming procedure.  The result is the following modified procedure: 

\medskip

\noindent{\bf{\mpr'} (informal):}\\
Given are $\delta$-viable sets $V_1 \subset \mH_A$ and $V_2 \subset \mH_B$ of dimension $q$, each specified by MPS with polynomial bond dimension.
\begin{enumerate}
\item {\bf Tensor} the two sets, as in Lemma \ref{l:t}, to obtain a $2\delta$-viable set  $W= V_1 \otimes V_2$ of dimension at most~$q^2$.
\item Perform {\bf viable set amplification}  followed by {\bf trimming} on the viable set to produce a $\delta$-viable set of smaller dimension, again specified by MPS with polynomial bond dimension.  Repeat this step until the resulting $\delta$-viable set has dimension $q$.
\end{enumerate}

We note that the correctness of the trimming procedure employed in the second step of \mpr' relies on the area law established using \mpr, as described in the previous section. 

The overview given in this section provides an accurate outline of how viable sets can be put together into an efficient algorithm for mapping out the low-energy subspace of a local Hamiltonian. The most important technical ingredient that we have set aside so far is the creation of AGSP with the required parameter trade-off between $D$ and $\Delta$. In Section~\ref{sec:arealaws} we establish \emph{existence} of the desired AGSP, which lets us formally implement the first part of our results, area laws for local Hamiltonians satisfying assumptions (DG) and (LD) described in the introduction. In order to obtain algorithms we will need to make the AGSP constructions \emph{efficient}: this is achieved in Section~\ref{sec:efficient}, with the resulting algorithms described in Section~\ref{sec:algorithms}.

\section{Area laws}
\label{sec:arealaws}

In this section we establish area laws for the ground space and low-energy space of Hamiltonians satisfying assumptions (DG) and (LD) respectively. The proofs are based on the non-constructive bootstrapping argument outlined in Section~\ref{sec:viable-amplification}, which relies on a sufficiently good construction of AGSP. We first review the general Chebyshev-based AGSP construction from~\cite{ref:AKLV2013-AL} in Section~\ref{sec:Cheby}. We introduce a scheme of hard truncation for the norm of a Hamiltonian in Section~\ref{sec:hard-truncation}. In section~\ref{sec:agsp-constructions} we apply the Chebyshev construction to the truncated Hamiltonian to obtain our main AGSP constructions. The AGSP are applied to the proof of the area law under assumption (DG) in Section~\ref{s:bootstrap-dg} and assumption (LD) in Section~\ref{s:bootstrap-ld}.

\subsection{The Chebyshev polynomial AGSP}
\label{sec:Cheby}

Given a Hamiltonian $H$ with ground energy $\eps_0$ and a gap parameter $\spg$, a natural way to define an approximate ground state
projection is by setting $K\EqDef P_k(H)$, where $P_k$ is a
polynomial that satisfies $P_k(\eps_0)=1$ and $|P_k(x)|^2\le \Delta$
for every $\eps_0+\spg\le x\le \norm{H}$. Clearly, $K$ preserves the
ground space and reduces the norm of any eigenstate $\ket{\phi}$ of $H$ with eigenvalue at least $\eps_0+\spg$ as $\norm{K\ket{\phi}}^2\le \Delta$. Moreover, 
the lower the degree of $P_k$, the lower the Schmidt rank of
$K$ at every cut. Following~\cite{ref:AKLV2013-AL} we
construct such a polynomial based on the use of Chebyshev
polynomials. The construction is summarized in the following
definition.

\begin{definition}[The Chebyshev-based AGSP]
\label{def:Cheby-AGSP}
  Let $H$ be a Hamiltonian and $\eta_0<\eta_1$ two parameters.\footnote{$\eta_0$ and $\eta_1$ may be chosen as the ground state energy and first excited energy of $H$ respectively, but they need not.} For any integer $k>0$, let $T_k$ be the $k$-th Chebyshev polynomial of the
  first kind, and $P_k$ the following rescaling of
  $T_k$:
  \begin{align}
  \label{def:Pn-Cheby}
    P_k(x) \EqDef \frac{1}{\tilde{P}_k(\eta_0)}\tilde{P}_k(x)\,,
    \quad\text{where} \quad
    \tilde{P}_k(x) 
      \EqDef T_k\left(2\frac{x-\eta_1}{\norm{H}-\eta_1}
        -1\right) \,.    
  \end{align}
  The \emph{Chebyshev AGSP of degree $k$} for $H$ is $K \EqDef P_k(H)$.
\end{definition}

The properties of the Chebyshev AGSP are given in the following
theorem. Here and throughout we use the convention that a 1D local Hamiltonian on $n$ qudits numbered $1,\ldots,n$ decomposes as $H= \sum_{i=1}^{n-1} h_i$, where $0\leq h_i \leq \Id$ is the local term acting on qudits $\{i,i+1\}$.

\begin{theorem}
\label{thm:Cheby-AGSP} Let $H$ be a Hamiltonian on $n$ qudits, $\eta_0<\eta_1$ two parameters and $\spg=\eta_1-\eta_0$. Suppose that for
  some $i_1<i_2\in \{1, \ldots, n\}$ and $3\leq \ell \leq (i_2-i_1)/2$, $H$ can be written
  as 
  \begin{align}
  \label{eq:H-mid-cut}
    &H_L + h_{i_1-\ell} + \ldots + h_{i_1} + \ldots + h_{i_1+\ell-1} \notag\\
		&\qquad+ H_M + h_{i_2-\ell} + \ldots + h_{i_2} + \ldots + h_{i_2+\ell-1} +
    H_R \,,	
  \end{align}
  where each $h_i$ is a 2-local operator on qudits $\{i,i+1\}$ and
  $H_L$, $H_M$ and $H_R$ are defined on qudits $J_L=\{1,\ldots, i_1-\ell\}$, $J_M=\{i_1+\ell,\ldots, i_2-\ell\}$ and
  $J_R=\{i_2+\ell, \ldots, n\}$ respectively. For any integer $k>0$ let
		$$  \Delta\EqDef 4e^{-4k\sqrt{\frac{\spg}{\norm{H}-\eta_0}}}.$$
	Then the degree-$k$ Chebyshev AGSP $K$ is a $(D,\Delta$) spectral AGSP for $(H,\eta_0,\eta_1)$ such that:
\begin{enumerate}
\item\label{item:preserve-ev} For any eigenvector $\ket{\psi}$ of $H$ with associated eigenvalue $\lambda$, $\ket{\psi}$ is an eigenvector of $K$ with associated eigenvalue $P_k(\lambda)$. 
\item\label{item:preserve-gs} If $\lambda \leq \eta_0$ then $P_k(\lambda) \geq 1$,  $P_k(\eta_0)=1$,  and if  $\lambda\leq \eta_0+\spg/k$ then 
$$ P_k(\lambda)\geq 1- O\Big(\frac{k|\lambda-\eta_0|}{\spg\|H\|} \sqrt{\Delta}\Big).$$
\item\label{item:orth-shrink} If $\lambda \geq \eta_1$ then $P_k(\lambda) \leq \sqrt{\Delta}$.
\item\label{item:large-sr} The Schmidt rank of $K$ at all cuts in the region $J_M$ (resp. $J_L$, $J_R$) satisfies $B \leq \tilde{B}^{O(k)}$, where $\tilde{B}$ is an upper bound on the
      Schmidt rank of $H_M$ (resp. $H_L$, $H_R$) at every cut. 
						  \item\label{item:small-sr} The Schmidt rank of $K$ with respect to the cuts $(i_1,i_1+1)$ and $(i_2,i_2+1)$ satisfies
      $D \leq (dk)^{\bigO{\ell + k/\ell}}$.
\end{enumerate}
\end{theorem}

\begin{proof}
  Item~\ref{item:preserve-ev}. follows from the definition of $K=T_k(H)$ as a polynomial in $H$ (see Definition~\ref{def:Cheby-AGSP}). 
	For item~\ref{item:preserve-gs}. and item.~\ref{item:orth-shrink} we use the following properties of $T_k$ (see e.g.~\cite{ref:AKLV2013-AL} and~\cite[Lemma~B.1]{ref:Kuwahara2015-LR} for a proof): 
	\begin{align}
	|T_k(x)|&\le 1& \text{for }|x|\le 1,\label{eq:cheb-1a}\\
	|T_k(x)|&\ge\frac{1}{2}\exp\left(2k\sqrt{\frac{|x|-1}{|x|+1}}\right)&\text{for }|x|\ge 1,
	\label{eq:cheb-1b}\\
	T_k(x) &= \frac{1}{2}\big(x+\sqrt{x^2-1}\big)^k + \frac{1}{2}\big(x-\sqrt{x^2-1}\big)^k&\text{for }|x|\ge 1.\label{eq:cheb-1c}
	\end{align}
	The fact that eigenvectors with eigenvalue $\eta_0$ are mapped to fixed points of $K$ follows from $P_k(\eta_0)=1$. Next suppose $\ket{\psi}$ is an eigenvector of $H$ with eigenvalue $\eta_0+\delta$ where $|\delta|<\eta_1-\eta_0$. From~\eqref{eq:cheb-1c} we see $|T_k(x+\delta)-T_k(x)| = O(k\delta/\min(x^2-1,x\pm\sqrt{x^2-1}))$ as long as $x,x+\delta\leq -1$. Taking into account the scaling used to define $P_k$, 
	$$|P_k(\eta_0+\delta)-P_k(\eta_0)| = O\Big( \frac{1}{\tilde{P}_k(\eta_0)} \frac{k\delta}{\spg \|H\|}\Big) = O\Big(\frac{\delta\,k}{\spg\|H\|}\Big)  e^{-2k \sqrt{\frac{\spg}{\|H\|-\eta_0}}},$$
	where the last inequality uses~\eqref{eq:cheb-1b}. Item~\ref{item:orth-shrink} follows by combining~\eqref{eq:cheb-1a} and~\eqref{eq:cheb-1b}. 
  
	Item~\ref{item:large-sr}. is immediate since $K$ is computed as a linear combination of $j$-th powers of $H$ for $j\leq k$. 
	
	Finally, for a proof of item~\ref{item:small-sr} we refer to Proposition~\ref{prop:Cheby-AGSP-eff} in Section~\ref{sec:cheby-efficient} below.
\end{proof}

Theorem~\ref{thm:Cheby-AGSP} provides us with a powerful recipe for
constructing good AGSP. To minimize the Schmidt rank at a cut $(i,i+1)$ for $i\in\{i_1,i_2\}$ we should take $k=\bTheta{\ell^2}$, which gives a bound of $D \leq (dk)^{O(\sqrt{k})}$, a much better bound than the naive
$d^{O(k)}$. To guarantee a small $\Delta$ we should choose $k$ large enough
to ensure that $e^{-4k\sqrt{\spg/\norm{H}}}$ remains small, which
requires the Hamiltonian to have a small norm. This is the
role of the truncation scheme presented in the
following section.

\subsection{Hard truncation}
\label{sec:hard-truncation}

We introduce a scheme of \emph{hard truncation} that is appropriate (though not efficient) for truncating the norm of an arbitrary local Hamiltonian in a certain region $J$, while preserving its low-energy eigenspace $H_{[\eps_0,\eps_0+\eta]}$. The basic idea is to replace $H\mapsto H \Pi_{\le \eps_0+t} +
(\eps_0+t) \Pi_{ >\eps_0+t}$, where $\Pi_{\le t}$ projects onto the span of eigenvectors of $H$ with eigenvalue less than $t$, $\Pi_{>\eps_0+t} \EqDef \Id-\Pi_{\le
\eps_0+t}$, and $t$ is chosen to be large enough with respect to $\eta$. 

\begin{definition}[Hard truncation]\label{def:hard-trunc}
Let $t>0$,  $H = H_J + H_{\overline{J}}$ where $H_J = h_{j_0} + h_{j_0+1} +
  \ldots + h_{j_1-1}$ is a local Hamiltonian acting on a contiguous set of  
  qudits $J= \{j_0, j_0+1,\ldots, j_1\}$, and
  let $\eps_J$ be the ground energy of $H_J$. Let $\Pi_-$ be the
  projector onto the span of all eigenvectors of $H_J$ with eigenvalue less than $\eps_J+t$,
  and $\Pi_+\EqDef \Id-\Pi_-$.
  Then the \emph{hard truncation} of $H_J$ is given by
  \begin{align} \label{e:ht}
    \tilde{H}_J \EqDef H_J \Pi_- + (t+\eps_J)\Pi_+   
  \end{align}
	and the \emph{hard-truncated Hamiltonian} $\tilde{H}_t$ associated to the region $J$ is 
	$$ \tilde{H}_t = \tilde{H}_J + H_{\overline{J}}.$$
\end{definition}

We now show that truncating a $n$-qubit Hamiltonian on a subset $J$ of the qubits leads to a truncated Hamiltonian whose low-energy space is close to that of the original Hamiltonian. The main tool in proving this result is
Theorem~2.6 of~\cite{AradKZ14energy}, a
generalization and strengthening of the truncation result that
appeared in~\cite{ref:AKLV2013-AL}. Adapted to the current
setting it can be formulated as follows. 

\begin{proposition}[Adapted from Theorem~2.6 in~\cite{AradKZ14energy}]
\label{prop:truncation}
  For any  $\eta>0$ let $\Pi_{\le \eta}$ denote the
  projector on the span of all eigenvectors of $H$ with eigenvalue at most $
  \eta$, and similarly $\tilde{\Pi}_{\le \eta}$ for $\tilde{H}_t$.
  Let
  $\eps_0\le\eps_1\le\eps_2\le\ldots$ and
  $\tilde{\eps}_0\le\tilde{\eps}_1\le\tilde{\eps}_2\ldots$ be the
  sorted eigenvalues of $H$ and $\tilde{H}_t$ respectively, where eigenvalues appear with multiplicity. For any $\eta>0$, let 
	\begin{equation}\label{eq:def-xi}
	\xi = e^{(t-\eta)/8+24}.
	\end{equation}
  Then the following hold:
  \begin{enumerate}
    \item $\norm{(H-\tilde{H}_t)\Pi_{\le\eps_0+\eta}} \le \xi$ and 
      $\norm{(H-\tilde{H}_t)\tilde{\Pi}_{\le\eps_0+\eta}} \le \xi$,
    \item For all $j$ for which $\eps_j\le \eps_0+\eta$, we have
      $\eps_j-\xi \le \tilde{\eps}_j \le \eps_j$.
  \end{enumerate}   
\end{proposition}

\begin{proof} The proposition follows from
  Theorem~2.6 in \Ref{AradKZ14energy} by using $\lambda=\frac{1}{8}$
  and the fact that $\eps_0\le \tilde{\eps}_0+2$ to bound
  $\Delta\tilde{\eps}$ by $\Delta\eps+2$. Here we can take $|\partial L|=2$
  since there are two boundary terms connecting the truncated region $J$
  and the rest of the system. 
    \end{proof}

The following lemma summarizes the approximation properties of the hard truncation procedure that will be important for us. 

\begin{lemma}
\label{lem:hard-trunc} 
 For any $\eta>0$, let $T_\eta =H_{[\eps_0,\eps_0+\eta]}$ be the low-energy eigenspace of $H$,  $J=\{j_0, \ldots, j_1\}$ a contiguous
  subset of qudits and $\tilde{H}_t$ the associated hard-truncated
  Hamiltonian, with corresponding low-energy eigenspace $\tilde{T}_\eta=\tilde{H}_{[\tilde{\eps}_0,\tilde{\eps}_0+\eta]}$.
	Let $\xi$ be as defined in~\eqref{eq:def-xi}. 
	Then the following hold for any $t>\eta$:
  \begin{enumerate}
    \item The ground energy $\tilde{\eps}_0$ of $\tilde{H}_t$ satisfies
      $\eps_0- C e^{-c(t-\eta)} \le \tilde{\eps}_0 \le \eps_0$ for some universal constants $C,c$.
      
    
    \item For any $\delta>0$ there is 
		$$\eta' = \eta+\sqrt{\frac{\eta}{\delta}}\, e^{-\Omega(t-\eta)}$$
		such that the subspace $\tilde{T}_{\eta'}$ is $\delta$-close to ${T}_\eta$, and $T_{\eta'}$ is $\delta$-close to $\tilde{T}_\eta$.

  \end{enumerate}
\end{lemma}

\begin{proof}
  The first item follows directly from  the second item in Proposition~\ref{prop:truncation}. 
	For the second item, we prove that $\tilde{T}_{\eta'}$ is $\delta$-close to ${T}_\eta$, the proof of the second relation being identical. Fix a small width parameter $h$ (to be specified later) and let $\ket{\psi} = \sum_i \beta_i \ket{\psi_i}$ be supported on eigenvectors $\ket{\psi_i}$ of $H$ with eigenvalue $\mu_i \in [\lambda-h,\lambda+h]$ with $\lambda \leq \eps_0+\eta$. Then $\|H\ket{\psi} - \lambda \ket{\psi}\| \leq h$. 
	Decompose $\ket{\psi} = \sum \alpha_i \ket{\phi_i}$, where for each $i$, $\ket{\phi_i}$ is an eigenvector of $\tilde{H}_t$ with associated eigenvalue $\tilde{\lambda}_i$. Using the first item in  Proposition~\ref{prop:truncation}, 
	\begin{align*}
\sum_i |\alpha_i|^2|\lambda -  \tilde{\lambda}_i|^2 &\leq \big(\| (H-\tilde{H})\ket{\psi}\|  + \|(H-\lambda\Id)\ket{\psi}\|\big)^2\\
	& \leq \big(e^{-\Omega(t-\eta)} + h\big)^2. 
	\end{align*}
By Markov's inequality it follows that for any $\delta>0$
	$$\big\|\tilde{\Pi}_{> \lambda+\delta}  \ket{\psi}\big\| \leq \frac{ e^{-\Omega(t-\eta)} + h}{\delta}.$$
 Any $\ket{\psi}$ in $T_\eta$ can be written as a linear combination $\ket{\psi} = \sum_j \beta_j \ket{h_j}$ with each $\ket{h_j}$ supported on eigenvectors of $H$ with eigenvalue in a small window of width $2h$, and the number of terms at most $\lceil \frac{\eta-\eps_0}{2h}\rceil$. Thus
\begin{align*}
\big\|\tilde{\Pi}_{> {\eps}_0+ \eta'}  \ket{\psi}\big\| &\leq \sum_j |\beta_j|\big\|\tilde{\Pi}_{> {\eps}_0+ \eta'}  \ket{h_j}\big\| \\
&\leq \sqrt{\frac{\eta}{2h}}\frac{ e^{-\Omega(t-\eta)} + h}{\eta'-\eta} .
\end{align*}
Choosing $h  = e^{-\Theta(t)}$, we see that the choice of $\eta'$ made in the statement of the lemma suffices to ensure that this quantity is at most $\sqrt{\delta}$, as desired. 
\end{proof}

\subsection{The AGSP constructions}
\label{sec:agsp-constructions}

The combination of Theorem \ref{thm:Cheby-AGSP}, Proposition \ref{prop:truncation}, and Lemma \ref{lem:hard-trunc} yield a construction that starts with a local Hamiltonian $H$, produces a truncated Hamiltonian $\tilde{H}$ with low energy space close to that of $H$ along with a spectral AGSP $K$ for $\tilde{H}$ with a good trade-off between the parameters $D$ and $\Delta$.

\begin{corollary} \label{c:goodagsphardtrunc}
Let $H$ be a 1D local Hamiltonian with ground energy $\eps_0$, and $\mH=\mH_L\otimes \mH_M\otimes \mH_R$ a decomposition of the $n$-qudit space in contiguous regions. For any integer $\ell\geq 1$ and $t>0$ there exists a Hamiltonian $\tilde{H}$ such that for any $\eps_0<\eta_0<\eta_1$ there is a $(D,\Delta)$ spectral AGSP $K$ for $(\tilde{H}, \eta_0, \eta_1)$ with the following properties.
\begin{enumerate}
\item $D = (d\ell)^{O(\ell)}$ and $\Delta = e^{-\Omega (\frac{\ell^2}{\sqrt{t+ \ell}}\sqrt{\eta_1-\eta_0})}$,
\item There are universal constants $C,c >0$ such that for $i\in\{0,1\}$
\begin{equation} \label{e:specclose}
0 \leq \eps_i - \tilde{\eps}_i \leq C e^{-c(t - \eps_0)}
\end{equation} 
where $\eps_i$, $\tilde{\eps}_i$ are the $i$-th smallest (counted with multiplicity) eigenvalues of $H$, $\tilde{H}$ respectively.
\item The space $H_{[\eps_0, \eta_1]}$ is $\delta$-close to $\tilde{H}_{[\tilde{\eps}_0, \eta_0]}$ and  $\tilde{H}_{[\tilde{\eps}_0, \eta_1]}$  is $\delta$-close to $H_{[\eps_0, \eta_0 ]}$, for 
\begin{equation} \label{e:setclose}
\delta = \Theta\Big(\frac{\eta_0- \eps_0}{(\eta_1-\eta_0)^2}\Big) e^{-\Omega(t- (\eta_0-\eps_0))}.
\end{equation}

\end{enumerate}
\end{corollary}

\begin{proof}
Let $L = \{1,\ldots,i_1\}$, $M=\{i_1+1,\ldots,i_2\}$ and $R=\{i_2+1,\ldots,n\}$ be the set of qudits contained in $\mH_L$, $\mH_M$ and $\mH_R$ respectively. We define the truncated Hamiltonian $\tilde{H}$ by applying the hard truncation transformation described in Definition~\ref{def:hard-trunc} thrice, to the regions $J_L = \{1,\ldots,i_1-\ell-1\}$, $J_M=\{i_1+\ell+1,\ldots,i_2-\ell-1\}$ and $J_R=\{i_2+\ell+1,\ldots,n\}$ respectively (provided each region is non-empty). The resulting truncated Hamiltonian $\tilde{H}=\tilde{H}_t$ has norm $O(\ell+t)$. 

Applying Lemma~\ref{lem:hard-trunc} thrice in sequence, for the three truncations performed, it follows that the sorted eigenvalues of $\tilde{H}$ satisfy~\eqref{e:specclose}. Eq.~\eqref{e:setclose} similarly follows from item 2. in Lemma~\ref{lem:hard-trunc} 

Finally we define the AGSP $K$ by applying the Chebyshev polynomial construction from Definition~\ref{def:Cheby-AGSP} to $\tilde{H}$ with a choice of $k=\ell^2$. The bounds on $\Delta$ and $D$ follow directly from item 3. and 5. from Theorem~\ref{thm:Cheby-AGSP} respectively. 
\end{proof}

From the corollary follow our two main AGSP contructions, which hold under assumptions (DG) and (LD) respectively. 

\begin{theorem}[Existence of AGSP, (DG)] \label{t:ais} Let $H$ be a local Hamiltonian satisfying Assumption (DG), and $\mH = \mH_L \otimes \mH_M \otimes \mH_R$ a decomposition of the $n$-qudit space in three contiguous blocks. There exists a collection of $D^2$ operators $\{A_i\}_{i=1}^{D^2}$ acting on $\mH_M$ along with a subspace $\tilde{T}\subseteq \mH$ such that: 
\begin{itemize}
\item $H_{[\eps_0,\eps_0+\eta_0]}$ and $\tilde{T}$ are mutually $.005$-close;
\item $D=\alone$,
\item There is $\Delta>0$ such that $D^{12} \Delta \leq \ddelta$ and whenever $S\subseteq \mH_M$ is $\delta$-viable for $\tilde{T}$ then $S'=\Span\{\cup_i A_i S\}$ is $\delta'$-viable for $\tilde{T}$, with $\delta'= \frac{\Delta}{(1-\delta)^2}$.
\end{itemize}
\end{theorem}

\begin{proof}
Let $\eta_0 = \eps_0 + \spg/10$ and $\eta_1 = \eps_0+9\spg/10$. Provided the implied constants are chosen large enough, setting $\ell = \Theta(\spg^{-1} \log \spg^{-1})$, $t=\tilde{O}(\ell)$ and $t> \tilde{O}(\frac{1}{\gamma} \log^2 (d/ \gamma))$ in Corollary~\ref{c:goodagsphardtrunc} gives $D^{12}\Delta < \ddelta$.  Due to the gap assumption it holds that $T=H_{[\eps_0, \eta_0]}= H_{[\eps_0, \eta_1]}$.  The choice of $t$ above  also ensures that the right-hand side of~(\ref{e:specclose}) is smaller than $\frac{1}{10}\gamma$ and the right hand side of~(\ref{e:setclose}) is smaller than $.005$, in which case $\tilde{H}$ has a spectral gap between $\eta_0$ and $\eta_1$, so that $\tilde{H}_{[\tilde{\eps}_0, \eta_0 ]}= \tilde{H}_{[\tilde{\eps}_0, \eta_1 ]}$. Then item 2 in the corollary implies that $\tilde{H}_{[\tilde{\eps}_0, \eta_0 ]}$ and $T$ are mutually $.005$-close, giving the first condition in the theorem with $\tilde{T}= \tilde{H}_{[\tilde{\eps}_0, \tilde{\eps}_0 + \frac{1}{10}\spg]}$.

The operators $\{A_i\}$ are defined from a decomposition $K= \sum_{i=1}^{D^2} L_i \otimes A_i \otimes R_i$ associated to the factorization $\mH=\mH_L\otimes\mH_M\otimes\mH_R$ of the AGSP from Corollary~\ref{c:goodagsphardtrunc}. Lemma~\ref{lem:error-spectral} gives the desired quantitative tradeoff between the increase in dimension of a viable set
and its increase in overlap, when acted upon by the $\{A_i\}$. 
\end{proof}

\begin{theorem}[Existence of AGSP, (LD)] \label{t:aisgapless} Let $\mu>0$ be a constant, $H$  a local Hamiltonian satisfying Assumption (LD), and $\mH = \mH_L \otimes \mH_M \otimes \mH_R$ a decomposition of the $n$-qudit space in three contiguous blocks. For any $\eta \geq \eta_1\geq 2\frac{\mu}{\log n}$ there exists a collection of $D^2$ operators $\{A_i\}_{i=1}^{D^2}$ acting on $\mH_M$ along with two subspaces $\tilde{T}_{-}\subseteq \tilde{T} \subseteq \mH$ such that:
\begin{itemize}
\item  $H_{[\eps_0, \eps_0 + \eta_1]}$ is $.005$-close to $\tilde{T}$, 
\item $\tilde{T}_{-}$  is  $.005$-close to $H_{[\eps_0, \eps_0 + \eta_1- \frac{\mu}{\log n}]} $,
\item  $D=\althree$,
\item There is a $\Delta>0$ such that $D^{12} \Delta < \ddelta$ and for any $S\subseteq \mH_M$ that is $\delta$-viable for $\tilde{T}$ it holds that $S'=\Span\{\cup_i A_i S\}$ is $\delta'$ -viable for $\tilde{T}_{-}$ with 
$\delta'= \frac{\Delta}{(1-\delta)^2}$.
\end{itemize}
\end{theorem}

\begin{proof}
The main difference with the proof of Theorem~\ref{t:ais} is that the parameter corresponding to the gap $\gamma$ is replaced by the quantity $\frac{\mu}{\log n}$. The proof of the first two items claimed in the theorem then closely mirrors that of Theorem~\ref{t:ais}.

It only remains to verify the third item. Despite having the desired AGSP, unlike in the gapped case we cannot hope to improve the quality of the viable set $S$ for all of $\tilde{T}=\tilde{H}_{[0, \eta' - \frac{\mu}{3\log n}]}$ by the application of the AGSP $K^k$.  However, if we view $S$ as a viable set for the smaller $\tilde{T}_{-}=\tilde{H}_{[0, \eta'  - \frac{2\mu}{3 \log n}]} \subseteq \tilde{T}$,  we now have an effective AGSP with respect to $\tilde{T}_{-}$ and the orthogonal complement of the larger $\tilde{T}$ and we can proceed as if in the presence of a small spectral gap of $ \frac{\mu}{3 \log n}$.  To see this, fix any vector $\ket{\psi} \in \tilde{T}_{-}$. Lemma \ref{l:2} shows that there exists a $\ket{w}\in S$ such that $\ket{w}= c \ket{\psi} + \ket{\psi^\perp}$ for some $\ket{\psi^\perp}$ orthogonal to $\tilde{T}$, and $c\geq (1-\delta)$. This brings us in line with the proof of  Lemma~\ref{lem:error-spectral} and we can use the same analysis to show that applying $K$ improves the parameter of the viable set $S$ from $\delta$ to the desired $\delta'= \frac{\Delta}{(1-\delta)^2}$. 
\end{proof}

\subsection{Area law for degenerate Hamiltonians} 
 \label{s:bootstrap-dg}

\begin{theorem}[Area law for degenerate gapped Hamiltonians]\label{thm:al-dg}
Let $H$ be a 1D local Hamiltonian acting on $n$ qudits of local dimension $d$ such that $H$ satisfies Assumption (DG). For any fixed cut and any $\delta=\poly^{-1}(n)$, for every unit $\ket{\psi}\in T$ there is an approximation $\ket{\psi'}$ such that $|\braket {\psi}{\psi'}|\geq 1- \delta$ and $\ket{\psi'}$ has Schmidt rank no larger than 
$$s(\delta)= r\, \aldgs$$
at that cut, and an MPS representation with bond dimension bounded by 
$$r \,\aldge.$$
Moreover, every state $\ket{\psi}\in T$ has entanglement entropy
\[ S(\ket{\psi}\bra{\psi}) \leq   \ln r + \tilde{O}\Big(\frac{1}{\gamma}\log^3 d\Big). \] 
\end{theorem}

The proof of the theorem proceeds in two steps. First we use a ``bootstrapping argument'' to show the existence of a viable set of constant error for the ground space, such that all states in the viable set have low Schmidt rank. The existence of arbitrarily good approximations with increasing Schmidt rank, as well as the bound on the entanglement entropy, follow by the application of a suitable AGSP. We state the bootstrapping step as the following proposition. The proposition can be understood as an analysis of the effect of a single application of the \mpr\ procedure introduced in Section~\ref{sec:merge-process} with the initial tensoring step omitted. (The connection will be made more formal once we analyze algorithms in Section~\ref{sec:algorithms}.) 

\begin{proposition}\label{prop:albootstrap-dg} 
Let $H$ be a local Hamiltonian satisfying assumption (DG), and $J\subseteq \{1,\ldots,n\}$. Then there exists a subspace  $W\subseteq \mH_J$ of dimension $q = r\alone$ that is $.015$-viable for the ground space $T$ of $H$.
\end{proposition}

\begin{proof}
Let  $W\subseteq \mH_J$ be a subspace of minimal dimension $q$ among all $.015$-viable subspaces for $T$. Let $\{A_i\}_{i=1}^{D^2}$ be AGSP operators guaranteed by Theorem~\ref{t:ais} for the Hamiltonian $H$ and region $M=J$, and $\tilde{T}$ the associated subspace. The first condition in the theorem together with Lemma~\ref{l:close} establishes that $W$ is $.04$-viable for $\tilde{T}$.
Let $s=q/(2D^2)$ and $W'\subseteq W$ a random subspace of dimension $s$. By Lemma~\ref{lem:viable-sample}, $W'$ is $(1-\delta')$-viable for $T$ with  $\delta'=s/(16q)=1/(32D^2)$ with positive probability provided 
\begin{equation}\label{eq:bootstrap-cond-dg}
s\,=\, \Omega\Big(\log q + r\log\Big(\frac{q}{s}\Big)\Big)
\end{equation}
for a large enough implied constant, as this will suffice to guarantee that $\eta$ stated in the lemma is strictly less than $1$. 

Let $S=  \cup_{i=1}^{D^2} A_i W'$. Then given our choice of $s$, $S$ has cardinality at most $q/2$, and by Lemma~\ref{lem:error-spectral} is $(32D^2)^2 \Delta$-viable for $\tilde{T}$. The condition $D^{12}\Delta < \ddelta$ implies $(32D^2)^2 \Delta + 0.005 \leq 0.1 \leq 0.15$, giving a contradiction with the minimality of $q$. The contradiction holds as long as the condition~\eqref{eq:bootstrap-cond-dg} holds, which given the bound on $D$ from Theorem~\ref{t:ais} will be the case as long as $q = r e^{\tilde{\Omega}(\gamma^{-1}\log^3 d)}$ for a large enough implied constant in the exponent.
 \end{proof}

Given the proposition, the proof of Theorem~\ref{thm:al-dg} follows by application of an AGSP derived from Corollary~\ref{c:goodagsphardtrunc}.

\begin{proof}[of Theorem~\ref{thm:al-dg}]
Fix a cut $\mH = \mH_L\otimes \mH_R$ as in the theorem. Let $V_L$ and $V_R$ be $0.015$-viable sets of minimal dimension for regions $J=L$ and $J=R$ respectively, and let $q$ be an upper bound on their dimension. Proposition~\ref{prop:albootstrap-dg} guarantees that we may take $q = r\alone$. By Lemma \ref{lem:merging} the set $W=V_L\otimes V_R$ is $.03$-viable for $T$.  The tensor product structure ensures that every element of $W$  has Schmidt rank no larger than $q$. Apply Corollary~\ref{c:goodagsphardtrunc} to $H$, with $\eta_0 = \eps_0+\spg/10$, $\eta_1 = \eps_0+9\spg/10$, $t = \Theta(\gamma^{-1/4}\log \delta^{-1})$ and $\ell= \Theta (t^{3/4})$. This gives a spectral AGSP $K$ with 
$$D=\altwo \quad\text{and}\quad \Delta \leq \delta/2,$$
for a Hamiltonian $\tilde{H}$ such that $\tilde{T}= \tilde{H}_{[\tilde{\eps}_0, \tilde{\eps}_0 + \frac{1}{10}\spg]}$ and $T$ are mutually  $(\delta/2)$-close.
Applying Lemma~\ref{lem:error-spectral} the set $W' = K W$ is $(\delta/2)$-viable for $\tilde{T}$ and every element within it has Schmidt rank no larger than $qD$.  Since $\tilde{T}$ and $T$ are $(\delta/2)$-close, $W'$ is $\delta$-viable for $T$. 

This proves the first statement in the theorem. The second follows by setting $\delta = \delta'/n$ in the above and noticing that the error made at each cut will add up linearly. The proof of the last statement is standard and follows from the bound on $s(\delta)$ as in~\cite{ref:AKLV2013-AL}:  we bound 
$$S\big(\ket{\psi}\bra{\psi}\big)\leq \ln \Big(r\alone\Big) + \sum_{i=3}^{\infty} 2^{-i} \log \big(s(2^{-(i+1)})\big),$$
which is dominated by the first term.   
\end{proof}

\subsection{Area law for low-density Hamiltonians} 
 \label{s:bootstrap-ld}

\begin{theorem}[Area law for low-density Hamiltonians]\label{thm:al-ld}
Let $H$ be a 1D local Hamiltonian  acting on $n$ qudits of local dimension $d$ such that $H$ satisfies Assumption (LD), $\mu< \eta \log n$ any positive constant and $T=H_{[\eps_0, \eps_0 + \eta - \mu/\log n]}$. For any fixed cut and any $\delta=\poly^{-1}(n)$, for every unit $\ket{\psi}\in T$ there is an approximation $\ket{\psi'}$ such that $|\braket {\psi}{\psi'}|>1- \delta$ and $\ket{\psi'}$ has Schmidt rank no larger than 
$$s(\delta)= r\, \allds$$
at that cut, and $\ket{\psi'}$ has an MPS representation with bond dimension bounded by 
$$r \,\allde.$$
Moreover, every state $\ket{\psi}\in T$ has entanglement entropy
\[ S(\ket{\psi}\bra{\psi}) \leq   \ln r + \tilde{O}\Big(\frac{\log n}{\mu}\log^3 d\Big) \]
\end{theorem}

As for Theorem~\ref{thm:al-dg}, the proof of Theorem~\ref{thm:al-ld} follows from a bootstrapping argument. We establish the analogue of Proposition~\ref{prop:albootstrap-dg} below. Just as for Theorem~\ref{thm:al-dg}, the theorem then follows by application of a suitable AGSP, and we omit that part of the proof. 

\begin{proposition}\label{prop:albootstrap-ld} 
Let $H$ be a local Hamiltonian satisfying assumption (LD), for some $\eta>0$. Let $J\subseteq \{1,\ldots,n\}$ and $\mu>0$. Then there exists a subspace  $W\subseteq \mH_J$ of dimension $q = r\althree$ that is $.015$-viable for the low-energy space $T_\mu = H_{[\eps_0,\eps_0+\eta-\mu/\log n]}$.
\end{proposition}

\begin{proof}
For fixed $d$ and $n$, let $C=C(d,n)$ be a constant such that the bound  $D\leq e^{C\frac{\log n}{\mu}\log^c (\frac{\log n}{\mu})}$ holds in Theorem~\ref{t:aisgapless} for all $\mu>0$, where $c>0$ is a universal constant implied by the $\tilde{O}$ notation. 
For any $\mu>0$ let $q(\mu)$ be the smallest dimension of a subspace $W_\mu\subset \mH_J$ that is $.015$-viable for $T_\mu$. Note that $q(\mu)$ is a non-increasing function of $\mu$. For $\mu>0$, let $r(\mu) = r e^{C' \frac{\log n}{\mu}\log(\log n/\mu)}$, where $C'=3C$. For any $\mu$, let $i_0$ be the smallest power of two such that $q(\mu/2^{i_0}) \leq r(\mu/2^{i_0})$. Note that $i_0$ is finite as $q(\mu) \leq d^n$ for all $\mu>0$. If $i_0 = 0$ then the proposition is proven. Suppose $i_0>0$, and let $\mu_0 = \mu/2^{i_0-1}$. Let $W=W_{\mu_0/2}$ be a subspace of dimension $q=q(\mu_0/2)$ that is $0.15$-viable for $T_{\mu_0/2}$. Let $\{A_i\}_{i=1}^{D^2}$ be AGSP operators guaranteed by Theorem~\ref{t:aisgapless} for the Hamiltonian $H$, region $M=J$, and parameters $\eta_1 = \eta - \mu_0 /(2\log n)$ and $\mu = \mu_0/2$. Let $\tilde{T}$ and $\tilde{T}_-$ be the resulting subspaces. The first condition in the theorem, together with Lemma~\ref{l:close}, establishes that $W$ is $.04$-viable for $\tilde{T}$.

Let $s= q(\mu_0)/D^2$ and $W'\subseteq W$ a random subspace of dimension $s$. By Lemma~\ref{lem:viable-sample} and the definition of $i_0$, $W'$ is $(1-\delta')$-viable for $\tilde{T}$ with  
\begin{equation}\label{eq:delta-bound-1}
\delta'=\frac{s}{16q}=\frac{q(\mu_0)}{16 q(\mu_0/2) D^2} \geq\frac{r(\mu_0)}{16 r(\mu_0/2) D^2} 
\end{equation}
with positive probability provided 
\begin{equation}\label{eq:bootstrap-cond-ld}
s\,=\, \Omega\Big(\log q + r\log\Big(\frac{q}{s}\Big)\Big)
\end{equation}
for a large enough implied constant, as this will suffice to guarantee that $\eta$ stated in the lemma is strictly less than $1$. 

Let $S=  \cup_{i=1}^{D^2} A_i W'$. Then $S$ has cardinality at most $q(\mu_0)$, and by Lemma \ref{lem:error-spectral} is $ \Delta/(\delta')^2$-viable for $\tilde{T}_{-}$, itself $0.005$-close to $W_{\mu_0}$. The condition $D^{12}\Delta < \ddelta$, together with~\eqref{eq:delta-bound-1} implies 
$$\frac{\Delta}{(\delta')^2} < 16^2 \cdot \ddelta D^{-8} \Big(\frac{r(\mu_0/2)}{r(\mu_0)}\Big)^2  \leq 10^{-2} e^{(2C'-8C) \frac{\log n}{\mu_0}\log^c\big(\frac{\log n}{\mu_0}\big)}.$$
Provided $C'\leq 4C$ this is at most $10^{-2}$, leading to a contradiction with the minimality of $i_0$. The contradiction holds as long as the condition~\eqref{eq:bootstrap-cond-ld} holds, which will be the case provided $C'>2C$. Choosing $C'=3C$ satisfies both conditions. 
 \end{proof}

\section{Efficient AGSP constructions}
\label{sec:efficient}


This section is devoted to the construction of efficiently computable, and efficiently implementable (as polynomial-size matrix product operators (MPO)), analogues of the existential AGSP constructions obtained in Section~\ref{sec:arealaws}. The first step in obtaining efficient constructions consists in replacing the method of hard truncation considered in Section~\ref{sec:hard-truncation} with a method of ``soft truncation''. This method, described in Section~\ref{sec:soft-truncation}, is somewhat less effective than hard truncation, but has the advantage that it can be made efficient; this is essential for its use in the algorithms presented in Section~\ref{sec:algorithms}. In Section~\ref{sec:cheby-efficient} we show that the Chebyshev polynomial AGSP introduced in Section~\ref{sec:Cheby} can also be made efficient. Our efficient  AGSP constructions for the (DG) and (LD) cases are provided in Section~\ref{sec:agsp-const-eff}. We conclude in Section~\ref{sec:FF-AGSP} with a more efficient construction specialized to the (FF) case; this last construction replaces the intricate AGSP constructions with a much simpler one based on the detectability lemma~\cite{ref:AharoAVZ2011-DL}. (The reader new to AGSP constructions may wish to start with the latter section.)

\subsection{Soft truncation}
\label{sec:soft-truncation}

We introduce a scheme of \emph{soft truncation} that reduces the norm of a local Hamiltonian $H$ in a certain region $J$ in a way that the truncated operator can be well-approximated by an MPO with small bond dimension. 
In hard truncation (Definition~\ref{def:hard-trunc}) the operator $\Pi_{\le \eps_0+t} H+
(\eps_0+t)\Pi_{> \eps_0+t})$ is used. This can be written
as $g_t(H)$, where $g_t(x)$ is defined by $g_t(x) \EqDef x$ for
$x\le \eps_0+t$ and $g_t(x)\EqDef t$ for $x>\eps_0+t$. The main idea of soft
truncation is to replace this non-smooth function by the infinitely
differentiable function
\begin{align}
  \label{def:ft}
 h_{t',t}(x) \EqDef t\Big(f_t(x)+ \frac{f_t(x)^2}{2} + \cdots + \frac{f_t(x)^{t'}}{t'}\Big),\quad \text{where}\quad f_t(x) \EqDef 1-e^{-x/t} \,,
\end{align}
which results in an operator $h_{t',t}(H)$ that closely approximates the hard-truncated Hamiltonian. Moreover, $h_{t',t}(H)$ can be given an efficient representation as an MPO by leveraging the truncated cluster expansion~\cite{hastings2006solving,KlieschGKRE14cluster} and its matrix product operator (MPO) representation from~\cite[Section~IV]{MolnarSVC15gibbsmpo}.

The following are basic properties of $h_{t',t}$.

\begin{lemma}
\label{lem:ft} 
  For any integers $t',t\geq 1$ and $x\ge 0$,
  \begin{equation*}
  \big| h_{t',t}(x)-x\big|\leq \frac{t}{t'}\Big(\frac{x}{t}\Big)^{t'},\qquad\text{and}\qquad \big|h_{t',t}(x)\big| \leq t \ln(t').
  \end{equation*}
\end{lemma}

\begin{proof}
Let $g_t(y) = -t\ln(1-y)$, so that $g_t(f_t(x)) = x$ for any $x\in [0,\infty)$. The function $h_{t',t}$ contains the first $t'$ terms of the Taylor expansion of $g_t$ around $0$, applied to $f_t(x)$, and the first inequality follows from Taylor's theorem and $f_t(x)\leq x$ for all $x$. The second inequality follows since $f_t(x)\leq 1$ for all $x$. 
\end{proof}

Recall our convention that a 1D local Hamiltonian acting on $n$ qudits numbered $1,\ldots,n$ decomposes as $H= \sum_{i=1}^{n-1} h_i$, where $0\leq h_i \leq \Id$ is the local term acting on qudits $\{i,i+1\}$. In addition to the truncation parameters $t$ and $t'$ the soft truncation construction is parametrized by a region $J\subseteq\{1,\ldots,n\}$ which specifies the set of local terms on which truncation is to be performed, and an energy $\eps'_J$ which is meant to be an approximation to the ground state energy of the restriction $H_J$ of $H$ to $J$. 

\begin{definition}[Soft truncation]
\label{def:soft-trunc} Let $H = H_J + H_{\overline{J}} $ be a 1D Hamiltonian, where $H_J = h_{j_0}+\cdots+h_{j_1-1}$ acts on a contiguous set $J=\{j_0,\ldots,j_1\}$ of qudits. Let $\eps_J$ be the ground energy of $H_J$, and $\eps'_J$ an approximation to $\eps_J$ satisfying $\eps_J-10\le \eps'_J\le \eps_J$. For given truncation parameters $t\geq t' \geq 1$, the \emph{soft truncation} of $H_J$
  is given by
  \begin{align*}
    \tilde{H}_J &\EqDef \eps'_J\Id + h_{t',t}(H_J - \eps'_J\Id),
  \end{align*}
  and the \emph{soft-truncated Hamiltonian} $H$ associated to region $J$ is
  \begin{align*}
    \tilde{H}_{t',t} \EqDef \tilde{H}_J+H_{\overline{J}}.
  \end{align*}
\end{definition}

The following lemma shows that for sufficiently large $t$ and $t'$, 
$\tilde{H}_{t',t}$ provides a good approximation to the lower part of the spectrum
of $H$.

\begin{lemma}
\label{lem:exp-approx}
Let $H = H_J + H_{\overline{J}}$ be a local 1D Hamiltonian. 
  Given truncation parameters $t\geq t' \geq 2$, the
  soft-truncated Hamiltonian $\tilde{H}_{t',t}$ satisfies $\tilde{H}_{t',t} \leq H$ and  for any eigenvector $\ket{\psi}$ of $H$ with energy $\lambda$ (resp. $\ket{\phi}$ of $\tilde{H}$ with energy $\mu\leq t$) it holds that 
	\begin{equation}\label{eq:exp-approx-0}
	\lambda- O\Big(\frac{(\lambda - \eps)^{t'}}{t't^{t'-1}}\Big) \leq \bra{\psi}\tilde{H}_{t',t}\ket{\psi} \leq \lambda \quad \text{and} \quad \mu \leq \bra{\phi}{H}\ket{\phi} \leq \mu + O\Big(\frac{(2(\mu - \eps))^{t'}}{t't^{t'-1}}\Big) ,
	\end{equation}
	where $\eps = \eps_{\overline{J}}+ \eps'_J$. In addition, if $H$ is gapped with gap $\spg$ then provided $t=\Omega(\spg^{-1})$, $\tilde{H}_{t',t}$ is gapped with gap $\spg/2\leq \tilde{\spg}\leq 2\spg$. 
	
	For $\eta>0$ let $T_\eta = H_{[\eps_0,\eps_0+\eta]}$ (resp. $\tilde{T}_\eta = \tilde{H}_{[\tilde{\eps}_0,\tilde{\eps}_0+\eta]}$) be the span of all eigenvectors of $H$ (resp. $\tilde{H}_{t',t}$) with associated eigenvalue in the indicated range. Then for any $\eta,\delta>0$ there is 
	$$\eta' \,=\, \eta+O\Big( \Big(\frac{\eta+10}{t}\Big)^{t'-1}\frac{1}{t'\sqrt{\delta}}\Big)$$
	such that the subspace  $\tilde{T}_{\eta'}$ is $\delta$-close to $T_\eta$ and $T_{\eta'}$ is $\delta$-close to $\tilde{T}_\eta$.

	\end{lemma}

\begin{proof}
 From Definition~\ref{def:soft-trunc},
  \begin{align}
    \tilde{H}_{t',t}-H = h_{t',t}(H_J-\eps'_J\Id)-(H_J-\eps'_J\Id)\,.\label{eq:exp-approx-1b}
  \end{align}
  Using the first bound from Lemma~\ref{lem:ft}, we get
  that for any vector $\ket{\psi}$,
  \begin{equation}\label{eq:exp-approx-1}
    \big|\bra{\psi}\tilde{H}_t\ket{\psi}-\bra{\psi}H\ket{\psi}\big|
      \le \frac{1}{t't^{t'-1}}\bra{\psi}(H_J-\eps_J'\Id)^{t'} \ket{\psi} \,.
  \end{equation}
Furthermore,
  \begin{align*}
    H_J-\eps_J'\Id       &\le H_{\overline{J}}-\eps_{\overline{J}}\Id + H_J-\eps'_J\Id \\
      &= H - (\eps_{\overline{J}} + \eps'_J)\Id, 
			\end{align*}
which combined with~\eqref{eq:exp-approx-1} and $H_J-\eps'_J\Id\geq 0$ proves the first two inequalities in~\eqref{eq:exp-approx-0}; the other two are obtained in the same way using in addition $x\leq 2h_{t',t}(x)$ for $0\leq x \leq t$. The relations between the spectral gaps of $H$ and $\tilde{H}_{t',t}$ follow from these inequalities.  

Starting from~\eqref{eq:exp-approx-1b}, squaring both sides and using (the square of) the first bound from Lemma~\ref{lem:ft} we get the operator inequality 
\begin{equation}\label{eq:exp-approx-3a}
(\tilde{H}_{t',t}-H)^2 \leq \frac{1}{(t')^2 t^{2t'-2}} (H_J-\eps'_J\Id)^{2t'}.
\end{equation}
Let $\bar{H}_{\overline{J}} = H_{\overline{J}} - h_{j_0-1} - h_{j_1}$, so that $\bar{H}_{\overline{J}}$ and $H_J$ commute. Using $\bar{H}_{\overline{J}}+(2-\eps_{\overline{J}})\Id \geq 0$, 
\begin{align}
(H_J-\eps'_J\Id)^{2t'} &\leq (H_J-\eps'_J\Id +\bar{H}_{\overline{J}}+(2-\eps_{\overline{J}})\Id  )^{2t'}\notag\\
&\leq ((H -\eps)\Id + 10\Id)^{2t'}.\label{eq:exp-approx-3b}
\end{align}
Let $\ket{\psi}$ be supported on eigenvectors of $H$ with eigenvalues in the range $[\lambda-h,\lambda+h]$ with $\lambda\leq \eps_0+\eta$ and $h$ a small width parameter. Decompose $\ket{\psi} = \sum_i \alpha_i \ket{\phi_i}$, where for each $i$, $\ket{\phi_i}$ is an eigenvector of $\tilde{H}_{t',t}$ with associated eigenvalue $\tilde{\lambda}_i$. Thus
\begin{align*}
\Big(\sum_i |\alpha_i|^2 |\lambda-\tilde{\lambda}_i|^2\Big)^{1/2} &\leq \Big\| \sum_i \alpha_i (\lambda-\tilde{\lambda}_i) \ket{\phi_i}\Big\|+h\\
 &= \big\| (\tilde{H}_{t',t}-H) \ket{\psi}\big\|+h\\
&= \bra{\psi}(\tilde{H}_{t',t}-H)^2 \ket{\psi}^{1/2}+h\\
&\leq \frac{1}{t' t^{t'-1}} \bra{\psi}((H -\eps)\Id + 10\Id)^{2t'}\ket{\psi}^{1/2}+h\\
&\leq \frac{1}{t' t^{t'-1}}(\eta+10)^{t'}+h,
\end{align*}
where the inequality before last follows by combining~\eqref{eq:exp-approx-3a} and~\eqref{eq:exp-approx-3b}. Applying Markov's inequality it follows that for any $\delta>0$
	$$\big\|\tilde{\Pi}_{> \lambda+\delta}  \ket{\psi}\big\| \leq \frac{ \frac{1}{t' t^{t'-1}}(\eta+10)^{t'} + h}{\delta}.$$
 Any $\ket{\psi}$ in $T_\eta$ can be written as a linear combination $\ket{\psi} = \sum_j \beta_j \ket{h_j}$ with each $\ket{h_j}$ supported on eigenvectors of $H$ with eigenvalue in a small window of width $2h$, and the number of terms is at most $\lceil \frac{\eta-\eps_0}{2h}\rceil$. Thus
\begin{align*}
\big\|\tilde{\Pi}_{> {\eps}_0+ \eta'}  \ket{\psi}\big\| &\leq \sum_j |\beta_j|\big\|\tilde{\Pi}_{> {\eps}_0+ \eta'}  \ket{h_j}\big\| \\
&\leq \sqrt{\frac{\eta}{2h}}\frac{ \frac{1}{t' t^{t'-1}}(\eta+10)^{t'} + h}{\eta'-\eta} .
\end{align*}
Chosing $h  = \frac{1}{t' t^{t'-1}}(\eta+10)^{t'}$, we see that the choice of $\eta'$ made in the statement of the lemma suffices to ensure that this quantity is at most $\delta$, as desired.

\end{proof}

We end this section by showing that the soft-truncated Hamiltonian
$\tilde{H}_{t',t}$ can be approximated by an operator with polynomial bond dimension which can be computed efficiently. Our construction is based on the cluster expansion
from~\cite{hastings2006solving,KlieschGKRE14cluster} in the 1D case, with some small
adjustments. We first state the result. 

\begin{lemma}\label{lem:exp-approx-eff}
	Let $t$ and $t'< (\ln(2)/2) t $ be truncation parameters and $H$ a $n$-qudit local Hamiltonian. For any $\xi>0 $ there is an MPO representation $\tilde{H}'$ for the truncated Hamiltonian $\tilde{H}=\tilde{H}_{t,t'}$ such that $\|\tilde{H}- \tilde{H}'\|\leq \xi$  and $\tilde{H}'$ has bond dimension $\poly(t'2^{t'} n/\xi)$ across all bonds. Such an MPO can be constructed in time polynomial in its size. 
\end{lemma}

\begin{proof}
The truncation $h_{t',t}(H)$ can be expressed as a linear combination of $O(t'2^{t'})$ terms of the form $e^{-\beta H}$ for values of $\beta$ in $\{1/t,\ldots,t'/t\}$; moreover the coefficients of the linear combination are at most $O(t'2^{t'})$ each. Using Theorem~\ref{thm:cluster-approx} and the assumption $t'/t\leq \ln(2)/2$ each $e^{-\beta H}$ can be approximated, in the operator norm, by an MPO of the form $M_r(H)$ with error less than $\xi/(t'2^{t'})^2$ as long as $r = \Omega( \ln((t')^22^{2t'}n^2/\xi))$. Finally, Theorem~\ref{thm:cluster} states that such an MPO with the claimed bond dimension can be found efficiently. 
\end{proof}

 Let $H=\sum_{i=1}^{n-1}h_i$ be a 1D, 2-local
Hamiltonian on $n$ qudits of dimension $d$, with $\norm{h_i}\le 1$
(but the $h_i$ are not necessarily non-negative), and let $\beta>0$
be an inverse temperature. We write the cluster expansion $e^{-\beta H}= \sum_w f(w)$,
where $w$ runs over all words on $\{1,\ldots,n-1\}$ and
$f(w)\EqDef\frac{(-\beta)^{|w|}}{|w|!} h_w$ with
$h_w\EqDef\prod_{i\in w} h_i$. For an integer $r>0$, let $S_{<r}$ be
the set of all those $w$ such that the support of $w$, the set of
qudits on which $h_w$ acts non-trivially, consists of connected
components of size smaller than $r$. Let $M_r(H)\EqDef\sum_{w\in
S_{<r}} f(w)$ be the ``truncated cluster expansion'' of $e^{-\beta
H}$. The following theorem follows from the proof of Lemma~2
in~\cite{KlieschGKRE14cluster}; we give the proof in Appendix~\ref{sec:cluster-mpo}.

\begin{theorem}\label{thm:cluster-approx}
  Let $\beta$ be such that $e^\beta-1<1$.  Then the following
  approximation holds in the operator norm:
  \begin{align*}
    \norm{e^{-\beta H} - M_r(H)} 
      \le e^{n^2 (e^\beta-1)^r}-1 \,.
  \end{align*}
\end{theorem}

The next theorem states that the operator $M_r(H)$ can be written
efficiently as an MPO. This encoding also
shows that the operator $M_r(H)$ has a low Schmidt rank. The proof,
which is given in Appendix~\ref{sec:cluster-mpo}, follows very closely the ideas of~\cite[Section~IV]{MolnarSVC15gibbsmpo}.

\begin{theorem}\label{thm:cluster}
  The $r^{\text{th}}$ order cluster expansion $M_r(H)$ of the
  operator $e^{-\beta H}$ can be written as an MPO of bond dimension
  $\le r^2d^r$ which can be computed in time $nd^{O(r)}$.
\end{theorem}

\subsection{The Chebyshev polynomial}
\label{sec:cheby-efficient}

For algorithmic purposes it is important that the Chebyshev AGSP can be constructed efficiently once one is given MPO representations for the truncated part of the Hamiltonian. The following proposition states that this is possible. 

\begin{proposition}
\label{prop:Cheby-AGSP-eff} 
Let $H$ be a Hamiltonian having a decomposition of the form described in~\eqref{eq:H-mid-cut}, $k$ an integer, and $K = P_k(H)$ the associated degree-$k$ Chebyshev AGSP as defined in Definition~\ref{def:Cheby-AGSP}. 
Assume that $H_M$ (but not necessarily $H_L$ or $H_R$) is specified by an MPO with bond dimensions at most $\tilde{B}$. 

Then there exists $D \leq (dk)^{\bigO{\ell + k/\ell}}$ such that a family of $D^2$ MPO $\{A_1,\ldots,A_{D^2}\}$ of bond dimension at most $\tilde{B}^k$ each such that there exists $B_1,\ldots,B_{D^2}$ with $K=\sum A_i\otimes B_i$ can be computed in time $nD^2\tilde{B}^{O(k)}$. Here the $A_i$ act on qudits $\{i_1,i_1+1,\ldots,i_2\}$ and the $B_i$ on the remaining qudits. 
 This computation does not require knowledge of $\eta_0,\eta_1$. 

Furthermore, if $H_L$ and $H_R$ are also given as MPO with bond dimension at most $\tilde{B}$ then the $B_i$ can be computed as well.
\end{proposition}

\begin{proof}
 The proof follows from a close examination of the proof of~\cite[Lemma~4.2]{ref:AKLV2013-AL}.\footnote{To follow the ensuing argument it may be helpful to translate the notation used for the indices in~\cite[Lemma~4.2]{ref:AKLV2013-AL} to the notation used here as follows: $s\to 2\ell-2$, $\ell\to k$, $k\to j$.}
		Adapting to our setting (where there are two cuts to consider simultaneously) the argument made in~\cite{ref:AKLV2013-AL} shows that in order to obtain an MPO for $K$ it suffices to include in the set $\{A_1,\ldots,A_{D^2}\}$ MPO representations for operators $P_{u_1u_2, kj_1j_2}(Z)$ where $u_1\in \{i_1-\ell,\ldots,i_1+\ell-1\}$, $u_2\in \{i_2-\ell,\ldots,i_2+\ell-1\}$, $j_1,j_2\in\{0,\ldots,k+2\ell-2\}$ and $Z$ is an $(4\ell-4)$-tuple of complex variables which takes on ${k-j_1+2\ell-2 \choose 2\ell-2}{k-j_2+2\ell-2 \choose 2\ell-2}$ possible values. For our purposes, a random choice of such values, e.g. distributed uniformly on the unit circle, will lead to a correct construction with probability $1$ (i.e. only depending on the number of digits of accuracy). We argue below that for each $P_{u_1u_2,kj_1j_2}(Z)$ one can efficiently construct an explicit set of MPO $\{A_\alpha\}$, where $1\leq i\leq \binom{k+j_1+1}{2j_1+1}\binom{k+j_2+1}{2j_2+1}d^{2(j_1+j_2) + 4\ell} $, such that there exists $B_i$ for which $\sum A_i\otimes B_i$ is an MPO for $P_{u_1u_2,kj_1j_2}(Z)$. This will lead to the claimed bounds as 
  \begin{align*}
  \sum_{u_1=i_1-\ell}^{i_1+\ell-1}  \sum_{u_2=i_2-\ell}^{i_2+\ell-1} \sum_{j_1,j_2=0}^{\lfloor k/2\ell\rfloor} d^{2(j_1+j_2) +4\ell}
&        \cdot \binom{k-j_1+2\ell-2}{2\ell-2}\binom{k-j_2+2\ell-2}{2\ell-2}\\
&\quad\cdot\binom{k+j_1+1}{2j_1+1}\binom{k+j_2+1}{2j_2+1}
  \end{align*}
	can be crudely bounded by $(dk)^{O(\ell+k/\ell)}$.
	
Fix $u_1,u_2$ and recall that $P_{u_1u_2,kj_1j_2}(Z)$ is defined as the sum of those terms in the expansion of $(H_L + \cdots + H_{i_1} + \cdots + H_{i_2} + \cdots +  H_R)^k$ which contain exactly $j_1$ (resp. $j_2$) occurrences of $H_{u_1}$ (resp. $H_{u_2}$).  There are $\binom{r+j_1+1}{2j_1+1}\binom{r+j_2+1}{2j_2+1}$ such terms. By cutting to the left of $u_1$ and right of $u_2$ we can efficiently construct at most $d^{2(j_1+j_2)}$ MPO which, properly combined, would give an MPO for the corresponding product. Finally we cut these MPO further so as to make the separation be to the left of $i_1$ and right of $i_2$ (or complete them appropriately, depending on whether $u_1\leq i_1$ or $u_1 > i_1$, and similarly for $u_2$ with respect to $i_2$). This last step multiplies the number of MPO by at most $d^{4\ell}$ (where we use $|i_1-u_1|,|i_2-u_2|\leq \ell$), giving the claimed bound. 
\end{proof}

\subsection{Efficient AGSP constructions}
\label{sec:agsp-const-eff}

We combine the soft truncation scheme with the Chebyshev polynomial AGSP to show that matrix product operator representations for operators $\{A_i\}$ satisfying the conditions of Theorem~\ref{t:ais} and Theorem~\ref{t:aisgapless} can be computed efficiently (in polynomial and quasi-polynomial time respectively). The same procedure, \gen, underlies both constructions, merely requiring a different choice of parameters in the two cases. The procedure is summarized in Figure~\ref{fig:generate} (it is implicit that the procedure is passed as an argument which assumption $H$ satisfies). We state its properties for the (DG) case in Theorem~\ref{t:aisefficient}, and for the (LD) case in Theorem~\ref{t:aisgaplessefficient}. For the case of a Hamiltonian satisfying assumption (DG) with a ground energy $\eps_0$ and a unique ground state (assumption (FF) of frustration-freeness) the procedure can be made even more efficient, and the result is stated in Theorem~\ref{t:aisff}.

\begin{figure}[H]
\rule[1ex]{12.2cm}{0.5pt}\\
\gen$(H, M, \eps'_M,(\eta_1,\mu))$: $H$ a Hamiltonian, $M=\{i_1+1,\ldots,i_2\}$ a subset of qudits, $\eps'_M$ an energy estimate for $H_M$, and $(\eta_1,\mu)$ energy parameters used only in the (LD) case. 
\begin{enumerate}
\item {\bf Soft truncation:} Set $\ell$ as in~\eqref{eq:deflff} in case (FF),~\eqref{eq:deflgapped} in case (DG), and~\eqref{eq:deflktgapped} in case (LD). Set $J_M$ as in~\eqref{eq:defjgapped}. In case (FF), construct an MPO for the truncated Hamiltonian as in Definition~\ref{def:FF-tH}. In case (DG) and (LD), construct an MPO for the soft-truncated Hamiltonian $\tilde{H}_M$ via the cluster expansion (see Definition~\ref{def:soft-trunc} and Lemma~\ref{lem:exp-approx-eff}).
\item {\bf Chebyshev polynomial:} Compute MPO representations for operators $\{A_i\}$ acting on $M$ using the decomposition of the Chebyshev polynomial provided in Proposition~\ref{prop:Cheby-AGSP-eff}, using energy parameters specified in~\eqref{eq:defetaff} in case (FF),~\eqref{eq:defetagapped} in case (DG), and~\eqref{eq:defetald} in case (LD).
\end{enumerate}
Return the MPO representations for $\tilde{H}_M$ and for the $\{A_i\}$. \vspace{0.2cm}\\
\rule[2ex]{12.2cm}{0.5pt}
\caption{The \gen\ procedure.}
\label{fig:generate}
\end{figure}

\begin{theorem}[Efficient AGSP, (DG)] \label{t:aisefficient} 
Let $H$ be a local Hamiltonian satisfying Assumption (DG), $\{1,\ldots,n\} = L\cup M\cup R$, where $L=\{1,\ldots,i_1\}$, $M=\{i_1+1,\ldots,i_2\}$, and $R=\{i_2+1,\ldots,n\}$, a partition of the $n$-qudit space, and ${\eps}'_M$ an estimate for the minimal  energy $\eps_M$ of the restriction of $H$ to $\mH_M$ such that $|\eps_M -{\eps'}_M |\leq 10$. Then the 
procedure $\gen (H, M, \eps'_M)$ described in Figure~\ref{fig:generate} returns
\begin{itemize}
\item MPO representations for a collection of $D^2$ operators $\{A_i\}_{i=1}^{D^2}$ acting on $\mH_M$ and of bond dimension at most  $n^{\tilde{O}(\spg^{-2})}$ such that there exists a subspace  $\tilde{T}$ for which the conclusions of Theorem~\ref{t:ais} are satisfied; 
\item An MPO for an operator $\tilde{H}_M$ such that $\norm{\tilde{H}_M} = O(\gamma^{-1} \log \gamma^{-1})$ and the minimal energy $\tilde{\eps}_M$ of $\tilde{H}_M$ restricted to $\tilde{T}$  satisfies $|\eps_M-\tilde{\eps}_M|<1/2$.
\end{itemize}
Moreover, $\gen (H, M, \eps'_M)$ runs in time
$n^{\tilde{O}(\spg^{-2})}$.\footnote{Here and in all our estimates
on running times we suppress dependence on the local dimension $d$,
which is treated as a constant.}
 \end{theorem}

\begin{proof}
We construct an AGSP $K$ from which the operators $\{A_i\}$ claimed in the theorem will be derived. The construction follows very closely the one employed in the proof of Theorem~\ref{t:ais}, replacing the use of hard truncation by soft truncation. 

The first step in \gen\ consists in truncating the Hamiltonian associated to each of the three regions. For this, introduce truncation parameters 
\begin{equation}\label{eq:deftgapped}
t=\Theta(\ell),\qquad t'=4,
\end{equation}
 a width parameter 
\begin{equation}\label{eq:deflgapped}
\ell = {\Theta}(\spg^{-1}\log \spg^{-1}),
\end{equation}
 and define a Hamiltonian $\tilde{H} = \tilde{H}_{t',t}$ by applying the soft truncation transformation described in Definition~\ref{def:soft-trunc} thrice, to the regions 
\begin{equation}\label{eq:defjgapped}
J_L = \{1,\ldots,i_1-\ell-1\},\quad J_M=\{i_1+\ell+1,\ldots,i_2-\ell-1\},\quad J_R=\{i_2+\ell+1,\ldots,n\}
\end{equation}
respectively (provided each region is non-empty). 
The resulting truncated Hamiltonian $\tilde{H}$ has norm $O(\ell+t\log t')=O(\ell)$. 
Note that the computation of the complete Hamiltonian $\tilde{H}$ requires estimates for the ground energies of the restriction of $H$ to each of the three regions that are being truncated. We will only need to efficiently compute an MPO for $\tilde{H}_M$, for which a rough estimate for the ground state energy of $H_M$, as provided as input to \gen, will be sufficient. 

The second step is to apply the Chebyshev polynomial from Definition~\ref{def:Cheby-AGSP} to $\tilde{H}$ to obtain the AGSP $K$. For this we make a choice of degree 
\begin{equation}\label{eq:defkgapped}
k = \ell^2
\end{equation}
and set the energy parameters $\eta_0$ and $\eta_1$ to
\begin{equation}\label{eq:defetagapped}
 \eta_0 = \eps_0 + \spg/10,\qquad \eta_1 = \eps_0 + 9\spg/10.
\end{equation}
We first verify that $K$ as defined is a spectral AGSP with the required properties, and then we show how it can be computed efficiently. 
By item 2. from Theorem~\ref{thm:Cheby-AGSP} the scaling parameter $\Delta$ is given by  
\begin{equation}\label{eq:gapless-delta-def}
\Delta \EqDef 4e^{-4k\sqrt{\frac{8\spg}{10(\|\tilde{H}\|-(\eps_0+\spg/10))}}}\,=\, e^{-\Omega\big( k \sqrt{\frac{\spg}{(\ell+t) }}\big)}.
\end{equation}
Furthermore, applying Theorem~\ref{thm:Cheby-AGSP} twice, once for the region centered at $i_1$ and once for the region centered at $i_2$, the bond parameter $D$ of $K$ across each of the cuts $(i_1:i_1+1)$ and $(i_2:i_2+1)$ is bounded by 
\begin{equation}\label{eq:gapless-D-def}
D \leq  (dk)^{O(\ell+k/\ell)} = \alone,
\end{equation} 
as desired. Moreover,
$$ D^{12}\Delta \,=\, e^{\spg^{-1}\tilde{O}(\log(\spg^{-1}))} e^{-\Omega(\spg^{-1}\log^{3/2}(\spg^{-1}))}$$
can be made smaller than $\ddelta$ by choosing the implicit constants appropriately. 

Next we apply Lemma~\ref{lem:exp-approx}  to evaluate the closeness between the low-energy subspaces of $H$ and $\tilde{H}$. Since $H$ has a spectral gap the subspace $T_{\spg/20} = H_{[\eps_0,\eps_0+\spg/20]}$ is the ground space $T$ of $H$. Setting $\delta = 0.05$ the lemma implies that $\tilde{H}_{[\tilde{\eps}_0,\tilde{\eps}_0+\spg/10]}$ is $\delta$-close to $T$ as long as the constant implied in the definition~\eqref{eq:deftgapped} of the truncation parameter $t$ is large enough. Conversely, we can write $T=T_{\spg/2} = H_{[\eps_0,\eps_0+9\spg/10]}$, in which case the lemma implies that $T$ is $\delta$-close to $\tilde{H}_{[\tilde{\eps_0},\tilde{\eps}_0+\spg/10]}$. Thus the two spaces are $\delta$-close. The claim on the ground state energies of $H_M$ and $\tilde{H}_M$ follows directly from Lemma~\ref{lem:exp-approx} and our choice of $t$. 

Finally we turn to efficiency, and verify that in time $n^{O(k)} = n^{\tilde{O}(\spg^{-2})}$  one can construct a set of at most $D^2$ MPO $A_1,\ldots,A_{D^2}$ acting on $\mH_M$ such that there exists $B_1,\ldots,B_{D^2}$ acting on $\mH_L\otimes \mH_R$  such that the AGSP $K$ can be represented as $K = \sum A_i \otimes B_i$. 
For this we first need to construct MPO representations for the truncated terms in the Hamiltonian. This is provided by Lemma~\ref{lem:exp-approx-eff} (applied to $H_M-\eps'_M\Id$), which given our choice of parameters $t,t'$  guarantees that an MPO providing inverse polynomial approximation (in the operator norm) to $\tilde{H}_M$ can be efficiently computed that has polynomial bond dimension across all cuts. Proposition~\ref{prop:Cheby-AGSP-eff} shows that an efficient construction of MPO for the $A_i$ follows. 
\end{proof}

\begin{theorem}[Efficient AGSP, (LD)] \label{t:aisgaplessefficient} 
Let $H$ be a local Hamiltonian satisfying Assumption (LD), parameters $\eta_1\leq \eta$ and  $\mu>0$, $\{1,\ldots,n\} = L\cup M\cup R$, where $L=\{1,\ldots,i_1\}$, $M=\{i_1+1,\ldots,i_2\}$, and $R=\{i_2+1,\ldots,n\}$, a partition of the $n$-qudit space, and ${\eps}'_M$ an estimate for the minimal  energy $\eps_M$ of the restriction of $H$ to $\mH_M$ such that $|\eps_M -{\eps'}_M |\leq 10$. Then the 
procedure $\gen (H, M, \eps'_M,(\eta_1,\mu))$ described in Figure~\ref{fig:generate} returns
\begin{itemize}
\item MPO representations for a collection of $D^2$ operators $\{A_i\}_{i=1}^{D^2}$ acting on $\mH_M$ and of bond dimension at most $e^{\tilde{O}(\log^3 n)}$ each such that there exists subspaces  $\tilde{T},\tilde{T}_-$ for which the conclusions of Theorem~\ref{t:aisgapless} are satisfied; 
\item An MPO for an operator $\tilde{H}_M$ such that $\norm{\tilde{H}_M} = \tilde{O}(\log(n)/\mu)$ and the minimal energy $\tilde{\eps}_M$ of $\tilde{H}_M$ restricted to $\tilde{T}_-$  satisfies $|\eps_M-\tilde{\eps}_M|<1/2$.
\end{itemize}
Moreover, $\gen (H, M, \eps'_M, (\eta_1,\mu))$ runs in time $e^{\tilde{O}(\log^3 n)}$.
\end{theorem}

\begin{proof}
The proof is similar to Theorem~\ref{t:aisefficient}, and the construction of $\tilde{H}$ and $K$ are the same except for a different choice of parameters. Here we choose 
\begin{equation}\label{eq:deflktgapped}
\ell = {\Theta}\Big( \frac{\log n}{\mu} \log \frac{\log  n}{\mu}\Big),\qquad k=\ell^2\qquad\text{and}\qquad t=\Theta(\ell),\qquad t'=4.
\end{equation}
The truncated Hamiltonian $\tilde{H}=\tilde{H}_{t,t'}$ is obtained as in the proof of Theorem~\ref{t:aisefficient}, by applying the soft truncation transformation described in Definition~\ref{def:soft-trunc} thrice. The AGSP $K$ is obtained by applying the Chebyshev polynomial from Definition~\ref{def:Cheby-AGSP} to $\tilde{H}$, with the  energy parameters $\eta'_0$ and $\eta'_1$ defined as
\begin{equation}\label{eq:defetald}
 \eta'_0 = {\eps}_0 + \eta_1 - \frac{\mu}{2\log n},\qquad \eta'_1 = {\eps}_0 + \eta_1
\end{equation}
respectively. As a result the parameters $D$ and $\Delta$ satisfy 
$$ D^{12}\Delta = e^{\frac{\log n}{\mu} \tilde{O}\big(\log  \frac{\log n}{\mu}\big)} e^{-\Omega\big( \frac{\log n}{\mu} \log^{1.5}( \frac{\log n}{\mu})\big)} = o(1),$$
which can be made less than $\ddelta$ by a proper choice of implied constants. The conditions on closeness of $T$, $T_-$ and $\tilde{T}$, $\tilde{T}_-$ follow from an application of Lemma~\ref{lem:exp-approx}, observing that our choice of truncation parameters $t,t'$ is sufficient to conclude closeness of the appropriate subspaces. The claim on the ground state energies of $H_M$ and $\tilde{H}_M$ follows directly from Lemma~\ref{lem:exp-approx} as well. 

Finally, applying Proposition~\ref{prop:Cheby-AGSP-eff} and Lemma~\ref{lem:exp-approx-eff} we see that an MPO for the part of $K$ acting on region $M$ can be computed in time $n^{O(k)} = e^{\tilde{O}(\log^2 n)}$. 
\end{proof}

\subsection{The frustration-free case}
\label{sec:FF-AGSP}

In this section we give a simpler construction of AGSP specialized to the case of a local Hamiltonian $H=\sum_i h_i$ that is frustration-free with a spectral gap $\spg>0$ and a unique ground state $\ket{\Gamma}$. Replacing each $h_i$ by the projection on its range preserves the
ground state and, given our usual normalization assumption $0\le h_i\le \Id$, can only increase the
spectral gap; thus we may without loss of generality assume that
each $h_i$ is a projection. 

We define a truncated version of $H$ based on the detectability
lemma from~\cite{ref:AharoAVZ2011-DL} as follows.

\begin{definition}[Truncated Hamiltonian in the frustration-free
  case] \label{def:FF-tH}  
  Suppose given a local Hamiltonian $H$ such that 
 $H = H_J + H_{\overline{J}}$ where $H_J = h_{j_0} + h_{j_0+1} +
  \ldots + h_{j_1-1}$ is a local Hamiltonian acting on a contiguous set of  
  qudits $J= \{j_0, j_0+1,\ldots, j_1 \}$.
	Let $J_e$ (resp. $J_o$) denote the subset of indices $i\in J$
  that are even (resp. odd). Define $H_{J,e}\EqDef\sum_{i\in J_e} h_i$ and
  $H_{J,o}\EqDef\sum_{j\in J_o} h_i$.
  Then the truncation of $H_J$ is given by
  $\tilde{H}_J \EqDef \tilde{H}_{J,e} + \tilde{H}_{J,o}$, where
  \begin{align}
    \label{def:HLeven-odd}
    \tilde{H}_{J,e} \EqDef \Id - \otimes_{i\in J_e}(\Id-h_i) \,, \qquad
    \tilde{H}_{J,o} \EqDef \Id - \otimes_{i\in J_o}(\Id-h_i)\,.
  \end{align}
The \emph{truncated Hamiltonian} $\tilde{H}$ associated to region $J$ is given by
  \begin{align}
    \label{eq:FF-tH}
    \tilde{H} \EqDef \tilde{H}_J + H_{\overline{J}}.
  \end{align}
\end{definition}

Clearly, $\tilde{H}_{J,e}$ and $\tilde{H}_{J,o}$ are projectors and
hence their norm is $1$. In addition, they are the sum of the
identity operator and a product of non-overlapping local terms, and
as such, their Schmidt rank is at most $d^2+1$ across any cut. We
show that $\tilde{H}$ has the same ground state as $H$, as well as a
large spectral gap.  This is done through the detectability lemma
and its converse stated below. 

\begin{definition}[The detectability lemma operator in 1D]
\label{def:DLO} 
  Let $H=h_1+\ldots+h_{n-1}$ be a 1D nearest-neighbor Hamiltonian such that each $h_i$ is a projector. Then the \emph{DL operator} of $H$ is defined by
  \begin{align*}
    \DL(H)\EqDef \otimes_i (\Id-h_{2i}) \otimes_i (\Id-h_{2i+1}).
  \end{align*}  
\end{definition}

Note that the operator $\DL(H)$ is in general not Hermitian. 
The usefulness of the definition comes primarily from the
detectability lemma:

\begin{lemma}[The detectability lemma]\label{lem:dl}
  Let $h_1, \ldots, h_m$ be projectors such that each
  $h_i$ commutes with all but at most $g$ other $h_j$, and
  let $H\EqDef\sum_i h_i$. 
  For any state $\ket{\psi}$ let $\ket{\phi}\EqDef
  \prod_i(\Id-h_i)\ket{\psi}$, where the product is taken in any order. Then
  \begin{equation}
    \norm{\ket{\phi}}^2 \le \frac{1}{\eps_\phi/g^2 + 1} \ ,\qquad\text{where}\qquad
    \eps_\phi \EqDef 
      \frac{1}{\norm{\ket{\phi}}^2}\bra{\phi}H\ket{\phi} \,.
  \end{equation}
\end{lemma}

The version of the detectability lemma stated above is stronger and
more general than the one appearing in~\cite{ref:AharoAVZ2011-DL}.
It also has a much simpler proof, which is given in~\cite{anshu2016simple}.
In addition to the detectability lemma, we will use a converse
statement which gives a lower bound on the norm of
$\DL(H)\ket{\psi}$. The converse, and its proof, appear in~\cite{anshu2016simple}.

\begin{lemma}[Converse of detectability lemma]
\label{lem:dl-converse}
  Let $H=\sum_{i=1}^{n-1}h_i$ be a 1D nearest-neighbor Hamiltonian
  such that each $h_i$ is a projector. Then for any eigenvector
  $\ket{\psi}$ of $H$,
  \begin{align}
    \norm{\DL(H)\ket{\psi}}^2 \ge 1-4\eps'_\psi \ ,\,
    \text{where}\qquad \eps'_\psi \EqDef \bra{\psi}H\ket{\psi} \,.
  \end{align}
\end{lemma}

With these two lemmas at hand we show the following.

\begin{theorem}
\label{thm:FF-tH} 
  The truncated Hamiltonian $\tilde{H}$ from
  Definition~\ref{def:FF-tH} satisfies the following:
  \begin{enumerate}
    \item $\tilde{H}$ is frustration free
      and has the same ground state $\ket{\gs}$ as $H$.    
    \item The Schmidt rank of $\tilde{H}$ at every cut is at most 
      $d^2+2$.
    \item $\tilde{H}$ has a spectral gap $\tilde{\spg}=\bOmega{\spg}$.
  \end{enumerate}
\end{theorem}
\begin{proof}
  Property~1. follows from the definition. For property~2. note
  first that the Schmidt rank of every operator on two
  $d$-dimensional qudits is at most $d^2$. This implies that the
  Schmidt rank of $\tilde{H}$ at every cut in $\overline{J}$ is at
  most $d^2+2$: we get a $d^2$ contribution from the local term that
  is defined on the cut and the extra $2$ comes from terms to the
  right/left of the cut. Consider now a cut between $i,i+1$ for an
  even $i$ that is in ${J}$. Since $i$ is even $\tilde{H}_{J,e}$ will contribute at most $d^2$, and $\tilde{H}_{J,o}$ at most $1$. The terms in $H_{\overline{J}}$ contribute at most $1$ as well, giving the claimed bound of $d^2+2$. 
    
  To prove~3. let $\ket{\psi}$ be orthogonal to
  $\ket{\gs}$. By the detectability lemma applied to $H$, 
  $\norm{\DL(H)\ket{\psi}}\le \frac{1}{\spg/4 + 1}$. By the converse of the detectability lemma applied to
  $\tilde{H}$, $\norm{\DL(\tilde{H})\ket{\psi}}\ge
  1-4\tilde{\spg}$. Since by construction $\DL(H)=\DL(\tilde{H})$, this implies
  $$\tilde{\spg}\ge \frac{1}{4}\Big(1-\frac{1}{\spg/4+1}\Big),$$
  from which the claim follows.
\end{proof}

The following is an analogue of Theorem~\ref{t:aisefficient} which provides a more efficient construction for the frustration-free case.

 \begin{theorem}[Efficient AGSP, (FF)] \label{t:aisff} 
Let $H$ be a local Hamiltonian satisfying Assumption (FF) and $\{1,\ldots,n\} = L\cup M\cup R$, where $L=\{1,\ldots,i_1\}$, $M=\{i_1+1,\ldots,i_2\}$, and $R=\{i_2+1,\ldots,n\}$ a partition of the $n$-qudit space. Then the procedure $\gen (H, M)$ returns MPO representations for a collection of $D^2$ operators $\{A_i\}_{i=1}^{D^2}$ acting on $\mH_M$ such that the following hold:
\begin{itemize}
\item $D=2^{\tilde{O} (\gamma ^{-1}\log^3  d)}$;
\item There is $\Delta>0$ such that $D^{12} \Delta < \ddelta$ and for any $S\subseteq \mH_M$ that is $\delta$-viable for $\{\ket{\Gamma}\}$ it holds that $S'=\Span\{\cup_i A_i S\}$ is $\delta'$-viable for $T$ with $\delta'= \frac{\Delta}{(1-\delta)^2}$;
\item Each $A_i$ has bond dimension at most $2^{\tilde{O}(\spg^{-2}\log^5 d)}$.
\end{itemize}
Moreover, for constant $d$ and $\spg>0$ the procedure $\gen (H, M)$ runs in time $n^{(1+o(1))}$.
\end{theorem}

\begin{proof}
We construct a suitable AGSP $K$ from which the operators $A_i$ will be derived.
The first step consists in truncating the Hamiltonian associated to each of the three regions. For this, introduce 
 a width parameter 
\begin{equation}\label{eq:deflff}
\ell = \tilde{\Theta}(\spg^{-1}\log^2 d),
\end{equation}
 and define a Hamiltonian $\tilde{H}$ by applying the truncation scheme described in Definition~\ref{def:FF-tH}   thrice, to the regions $J_L = \{1,\ldots,i_1-\ell-1\}$, $J_M=\{i_1+\ell+1,\ldots,i_2-\ell-1\}$ and $J_R=\{i_2+\ell+1,\ldots,n\}$ respectively (provided each region is non-empty). 
Based on Theorem~\ref{thm:FF-tH} the resulting truncated Hamiltonian $\tilde{H}$ has norm $O(1)$, the same ground state as $H$, and a spectral gap $\tilde{\spg} = \Theta(\spg)$. 

$K$ is obtained by applying Definition~\ref{def:Cheby-AGSP} to $\tilde{H}$ with 
\begin{equation}\label{eq:defetaff}
\eta_0=0,\qquad \eta_1=\tilde{\spg}
\end{equation}
and $k=\Theta(\ell^2)$. The bound on $D$ follows from Theorem~\ref{thm:Cheby-AGSP}, using which one can also verify that the desired trade-off $D^{12}\Delta < \ddelta$ will be achieved provided the right choice of constants is made in the choice of $\ell$. 

By Theorem~\ref{thm:FF-tH}  $\tilde{H}$ can be represented as an MPO with bond dimension $O(d^2)$, from which it follows that we can compute a decomposition $K = \sum L_i \otimes A_i \otimes R_i$ where each $A_i$ has bond dimension $O(d^k)= e^{\tilde{\Theta}(\gamma^{-2} \log^5 d)}$. 

The claim on the running time follows from the estimates provided in Proposition~\ref{prop:Cheby-AGSP-eff}.
\end{proof}

\section{Algorithms}
\label{sec:algorithms}

Equipped with the efficient construction of AGSP described in Section~\ref{sec:efficient}, we are  ready to turn \mpr\ into an efficient algorithm. The algorithm,  \low, follows the outline given in Section~\ref{sec:merge-process}, but requires additional ingredients. The first is the use of the procedure \gen\ described in Figure~\ref{fig:generate}, which creates MPO representations for the spectral AGSP required to perform error reduction. The second is an additional step of \emph{energy estimation}, which computes an energy estimate required by \gen.  

The complete algorithm is described in Figure~\ref{fig:algo1}. It takes as input a local Hamiltonian satisfying assumptions (FF), (DG) or (LD) (we assume the algorithm is told which assumption holds) and a precision parameter $\delta$, and returns MPS representations for a viable set that is $\delta$-close to the low-energy space $T$ of $H$.\footnote{The algorithm should also be provided a lower estimate for the gap $\spg$. If not, it can iterate for different values and return the lowest-energy states found.} 

\begin{figure}[H]
\rule[1ex]{12.2cm}{0.5pt}\\
\low $(H,\delta,(\eta,\mu))$: $H$ a local Hamiltonian acting on $\otimes_{i=1}^n \CC^d$, $n$ a power of two; $\delta$ an accuracy parameter; $(\eta,\mu)$ energy parameters for the (LD) case.
\begin{enumerate}
\item {\bf Initialization:} For $j\in\{1,2, \ldots n\}$ let $V^0_j$ contain a family of MPS representations for an (arbitrary) basis of  $\CC^d$, and $\eps_{0,j}'=0$. 
\item {\bf Iteration:} For $i$ from $1$ to $\log n$ do: \\[3mm]
For all $j\in\{1 ,\ldots, n/2^i\}$ do:
\begin{itemize} 
\item {\bf Generate:}    Let $M=\{(j-1)2^i, (j-1)2^i +1 , \dots ,  j2^i-1\}$ and $\eps_M' = \eps_{i-1,2j-1}' + \eps_{i-1, 2j}'$.\\
 Set $(\{A_i\},\tilde{H}_M)$= \gen $(H, M, \eps_M')$ in the (FF) and (DG) cases, and $(\{A_i\},\tilde{H}_M)$= \gen $(H, M, \eps_M',\eta-(i-1)\mu/\log n,\mu)$ in the (LD) case.
\item {\bf Merge:} Set $V^i_j =\mpr'( V^{i-1}_{2j-1}, V^{i-1}_{2j}, \{A_i\},s,(k,\xi))\subseteq \mathcal{H}_{[ (j-1)2^i +1,j2^i]} $, where $(s,k)$ are specified in~\eqref{eq:algo-sk-dg} and $\xi$ should satisfy~\eqref{eq:algo-xi-dg} for the case (DG) and (FF); in case (LD) the procedure \mpr\ can be used instead.
\item {\bf New Energy Estimation:} Form the subspace $V= \{A_i\}^{t} \cdot ( V^{i-1}_{2j-1}\otimes V^{i-1}_{2j})$, where $t = 4 \lceil \log \spg^{-1} \rceil$. Compute the smallest eigenvalue $\eps'_{i,j}$ of the restriction of $\tilde{H}_M$ to $V$. (This step is not needed in case (FF).)
\end{itemize}
\item[3.] {\bf Final step:}
Set $K=(\Id-H/\|H\|)$ and $\tau=10\|H\|\spg^{-1}\log(1/\delta)$. Choose an orthonormal basis $\{\ket{y_i^{(0)}}\}$ for $V^{\log n}_1$. Repeat for $t=1,\ldots,\tau$:
\begin{itemize}
\item Set $\{\ket{y_i^{(t)}}\} = \Trim_\xi( \Span\{K\ket{y_i^{(t-1)}}\} )$,
\end{itemize}
where $\xi$ is as previously in cases (DG) and (FF), and as in~\eqref{eq:algo-xi-ld} in case (LD).
\end{enumerate}
Return $\{\ket{z_i}\}$, the smallest $r$ eigenvectors  of $H$ restricted to $W=\Span\{\ket{y_i^{(\tau)}}\}$. \vspace{0.2cm}\\
\rule[2ex]{12.2cm}{0.5pt}
\caption{The \low\ algorithm.}
\label{fig:algo1}
\end{figure}

We note that the \low\ algorithm described in Figure~\ref{fig:algo1} already incorporates the modified procedure \mpr' sketched in Section~\ref{sec:merge-process}. As described in that section, \mpr' differs from \mpr\ by adding a step of bond trimming. The reason for the modification is that due to the logarithmic number of iterations, successive applications of \mpr\ may,  even if the $\{A_i\}$ can be applied efficiently, lead to MPS whose bond dimension eventually becomes super-polynomial. The procedure \mpr' is described and analyzed in detail in Section~\ref{sec:merge-prime}. In Section~\ref{sec:alg-dg}, Section~\ref{sec:alg-ff} and Section~\ref{sec:alg-ld} we build on the analysis of \mpr' and the efficient AGSP constructions from the previous section to show that \low\ leads to an efficient algorithm under assumptions (DG), (FF) and (LD) respectively.

\subsection{A modified \mpr\ procedure}
\label{sec:merge-prime}

The procedure \mpr' is described in Figure~\ref{fig:merge-prime}. It takes additional trimming parameters $k$ and $\xi$ as input ($k$ and $\xi$ will usually be of order $\log(n)$ and $\poly^{-1}(n)$ respectively).

\begin{figure}[H]
\rule[1ex]{12.2cm}{0.5pt}\\
\mpr' $(V_1, V_2, \{A_i\}, s, (k, \xi))$: Subsets $V_1\subseteq \mH_A$, $V_2\subseteq \mH_B$ of vectors (represented as MPS), operators $A_i$ acting on $\mH_1 \otimes \mH_2$ (represented as MPO), $s$ a dimension bound, $k\in\N$ and $\xi>0$ parameters for the trimming subroutine. 
\begin{enumerate}
\item {\bf Tensoring:} Set $W$ to be a set of MPS representations for an orthonormal basis for the space $\Span\{ V_1 \otimes V_2\}$.
\item {\bf Random Sampling:} Let $W'\subseteq W$ be a random $s$-dimensional subspace of $W$ obtained by applying a random orthogonal transformation to the vectors in $W$ and returning the first $s$ vectors obtained.
\item {\bf Error Reduction:} Set $V= W'$. Repeat $k$ times:
\begin{itemize}
\item  Set $V= \Trim_{\xi}(\Span\{\cup_{i} A_i W'\})$, where the trimming procedure $\Trim$ is described in Definition~\ref{def:trimming}. 
\end{itemize}
\end{enumerate}
Return MPS representations for the vectors in $V$. \vspace{0.2cm}\\
\rule[2ex]{12.2cm}{0.5pt}
\caption{The \mpr' procedure.}
\label{fig:merge-prime}
\end{figure}

Correctness of \mpr' (for an appropriate choice of $\xi$) relies on the area laws proven in Section~\ref{sec:arealaws} and on the analysis of the trimming procedure given in Section~\ref{ss:bt}. We give the analysis for the case of Hamiltonians satisfying assumption (DG) in the next section, for frustration-free Hamiltonians in Section~\ref{sec:alg-ff}, and for Hamiltonians satisfying assumption (LD) in Section~\ref{sec:alg-ff}.

\subsection{Degenerate Hamiltonians}
\label{sec:alg-dg}

The following theorem proves the correctness of algorithm \low\ for the case where the input Hamiltonian satisfies assumption (DG).  

\begin{theorem}  \label{t:alg1} Let $H$ be a local Hamiltonian satisfying Assumption (DG), $T$ its ground space, $r=\dim(T)$ and $\delta\geq \poly^{-1}(n)$. Then with probability at least $1- \frac{1}{n}$ the set of MPS returned by $\low(H,\delta)$ is $\delta$-viable for $T$.\footnote{The probability of success can be improved to $1  - \poly^{-1} (n)$ by scaling the parameter $s$ used in the algorithm by an appropriate constant.} The running time of the algorithm is  $n^{\tilde{O}(\spg^{-2})}$.
\end{theorem}

\begin{proof} 
The proof is based on the same ingredients as the proof of the area law given in Theorem~\ref{thm:al-dg}. There are two main differences:  we must show that the addition of the trimming step in \mpr' does not affect the quality of the viable set returned, and we must verify that the energy estimation step is sufficiently accurate. 

We show by induction on $i=0,\ldots,\log n$ that for all $j\in \{1,\ldots,n/2^i\}$, (i) the set  
  $V^{i-1}_{j}$ is $.015$-viable for $T$ and satisfies 
	 $|V^{i-1}_{j}|\leq Ds^2$, for $D$ and $s$ to be specified below, and (ii) $ \eps_{i-1,j}'$ is within an additive $\pm 3$ of its true value (the ground state energy of the restriction of $H$ to the corresponding spaces). 
	
	Both conditions are satisfied for $i=0$: for each $j\in\{1,\ldots,n\}$, 
  $V^0_j$ is $0$-viable for $T$ with $|V^0_j|=d$, and the energy estimate is accurate since the restriction of the Hamiltonian to a single qudit is identically $0$. 
	
	Suppose the induction hypothesis verified for $i-1$, fix $j\in \{1,\ldots,n/2^i\}$, and let $M$ be the region defined in the algorithm. 
		Correctness of the energy estimates $ \eps_{i-1,2j-1}'$ and $ \eps_{i-1,2j}'$ at step $(i-1)$ implies that $\eps_M '$ is within $\pm 7$ of the correct value $\eps_M$. By Theorem~\ref{t:aisefficient},  \gen \ returns a set of $D^2$ operators $\{A_i\}$ with the properties stated in Theorem~\ref{t:ais}. 

At this stage we are exactly in the same setting as for the proof of Proposition~\ref{prop:albootstrap-dg}, except for the additional trimming step in \mpr'. Following that proof we conclude that, prior to the trimming step, the merged set $V^i_j$ is $.01$-viable for $T$ with probability  $1- e^{-\Omega(s)}\geq 1-\frac{1}{n^2}$ provided $s = \Omega(r\log (q/s)+\log n)$. We choose
\begin{equation}\label{eq:algo-sk-dg}
s\geq 1600 r(\log r +1)\qquad\text{and}\qquad k = \frac{1}{2}\lceil \log_D(s) \rceil.
\end{equation}
This choice of $k$ ensures $s^2 \leq D^{2k} s \leq D s^2$, so that the bound on the dimension of $V^i_j$ required to establish the induction hypothesis holds. 

It remains to verify the quality of $V^i_j$ as a viable set. Note first that Theorem~\ref{thm:al-dg} allows us to bound the bond dimension $b$ of any vector in $T$ by a polynomial, at the expense of replacing $T$ by a set that is $10^{-4}$-close to $T$. Then the analysis given in Lemma~\ref{lem:trimming} shows that the effect of the trimming can be incorporated by replacing the error reduction parameter $\Delta$ associated with the $\{A_i\}$ by $(\Delta +\sqrt{nrb}\xi)$. Choosing $\xi$ such that 
\begin{equation}\label{eq:algo-xi-dg}
\sqrt{nrb}\xi < 10^{-4} \delta D^{-12},
\end{equation}
the remaining calculation applies and yields that $V^i_j$ is $.015$-viable for $T$.

Once this has been established, an application of the third item from Theorem~\ref{t:ais} shows that given the choice of $t$ made in the algorithm the subspace $V$ obtained after the energy estimation step is $O(\spg^{2})$-viable for $\tilde{T}$. Using that $\|\tilde{H}_M\| = O(\spg^{-1}\log \spg^{-1})$ it follows that $\eps'_{i,j}$ is within an arbitrarily small constant of the minimal energy of $\tilde{H}_M$ restricted to $\tilde{T}$.  Using the guarantee from Theorem~\ref{t:aisefficient},  $\eps'_{i,j}$ is within $\frac{3}{2}$ of  the minimal energy $\eps_M$ of the restriction of $H$ to $\mH_M$. This completes the inductive step.

We have shown that the iterative step succeeds with probability at least $1- 1/n^2$; since there are a total of $n$ such merging steps, applying a union bound the set $V^{\log n}_1$ is $.015$-viable with probability at least $1- \frac{1}{n}$.

To conclude it remains to analyze the final error improvement step. 
Let $\ket{\psi}$ be an eigenvector of $H$ with eigenvalue $\eps_0$, and $\ket{v}\in V_1^{\log n}$ such that $\ket{v} = \alpha \ket{\psi} + \sqrt{1-|\alpha|^2} \ket{v^\perp}$, where $\alpha \geq 0.9$ and $\ket{v^\perp}$ is supported on eigenvectors of $H$ with eigenvalue at least $\eps+\spg$. Following the same analysis as given in the proof of Lemma~\ref{lem:error-spectral} it follows that after renormalization the overlap of $K\ket{v}/\|K\ket{v}\|$ with $\ket{v}$ has improved from $\alpha$ to 
$$ \frac{\alpha^2}{\alpha^2 + (1-\alpha^2)(1-\spg/\|H\|)} =  \frac{\alpha^2}{1-\spg(1-\alpha^2)/\|H\|}\geq \alpha^2\Big(1+\frac{\spg}{2\|H\|}\Big).$$
Thus the set $K\{\ket{y_i^{(1)}}\}$ is $0.9(1+\spg/(2\|H\|))$-viable for $T$. Assuming $\xi$ is chosen small enough (satisfying~\eqref{eq:algo-xi-dg} suffices), by Lemma~\ref{lem:trimming} the set $\{\ket{y_i^{(2)}}\}$ will remain $0.9(1+\spg/(3\|H\|))$-viable for $T$. Repeating this procedure $\tau$ times yields a set $W$ that is $\delta$-viable for $T$.
Finally, each of the $r$ vectors $\ket{z_i}$ returned by the algorithm must have energy at most $\eps_0+\delta\spg$, which using the spectral gap condition implies that $\Span\{\ket{z_i}\}$ and $T$ are mutually $\delta$-close.   

The algorithm requires only a polynomial number of operations on MPS representations of vectors.  Due to trimming, all these vectors have polynomial bond dimension and thus each operation can be implemented in polynomial time. The complexity is dominated by the complexity of the procedure \gen\ and the application of the operators $A_i$, which is  $n^{\tilde{O}(\spg^{-2})}$.
\end{proof}

\subsection{Frustration-free Hamiltonians with a unique ground state}
\label{sec:alg-ff}

The most computation-intensive step of the \low\ algorithm is the construction, via \gen, and subsequent application in \mpr', of the set of operators $\{A_i\}$. In the special case where the Hamiltonian $H$ satisfies Assumption (FF), i.e. $H$ is frustration-free and has a spectral gap, the operators $\{A_i\}$  can be constructed very efficiently, yielding improved bounds on the running time. 
The overall algorithm remains as described in Figure~\ref{fig:algo1}, with \gen\ instantiated with the efficient procedure described in Theorem~\ref{t:aisff}. 

\begin{theorem} \label{t:unique}  Let $H$ be a local Hamiltonian satisfying Assumption (FF), $\ket{\Gamma}$ the unique ground state of $H$, and $\delta= n^{-\omega(1)}$.  With probability at least $1- \frac{1}{n}$ the lowest-energy vector $\ket{z}$ returned by \low$(H,\delta)$ satisfies $|\braket {z}{\Gamma}|\geq 1 -\delta$.  Moreover the algorithm runs in time  ${O}(n^{1+o(1)}\Time(n))$, where $\Time(n) = O(n^{2.38})$ denotes matrix multiplication time. 
\end{theorem}

\begin{proof} 
The proof follows very closely the proof of Theorem~\ref{t:alg1}, and we only indicate the main differences. To ensure the algorithm is efficient, it is
important to choose the trimming parameter $\xi$ to be as large as
possible. It follows from the area law for 1D gapped
systems~\cite{ref:AKLV2013-AL} (see also Theorem~\ref{thm:al-dg} for $r=1$) that the ground state
$\ket{\Gamma}$ of $H$ can be approximated up to accuracy
$\poly^{-1}(n)$ by a matrix product state with sub-linear bond
dimension.  Thus by Lemma \ref{lem:trimming}, using that $r,s$ are
both constant, and treating $d,\spg$ as constants, a choice of 
$\xi= n^{-(1/2 + \omega(1))}$ satisfying~\eqref{eq:algo-xi-dg}
 will suffice
to ensure the error remains negligible, while also maintaining the
property that all MPS manipulated have quasi-linear bond
dimension. The essential operations on such vectors required in the
algorithm, such as multiplication by an MPO $A_i$ of constant bond
dimension, or writing in canonical form, can all be computed in time
$\tilde{O}(n \Time(B))$ where $\Time(B)$ is matrix multiplication
time for $B\times B$ matrices and $B$ is an upper bound on the bond
dimension of the MPS being manipulated; $\Time(B)$ corresponds to
the cost of performing individual singular value decompositions on
the tensors that form each of the MPS. The claim on the running time
follows since the number of iterations of the algorithm is
logarithmic.  
\end{proof}

\subsection{Gapless Hamiltonians}
\label{sec:alg-ld}

We extend the analysis of the \low\ algorithm to the case of gapless Hamiltonians satisfying the (LD) assumption. The main obstacle, of course, consists in dealing with a gapless system. 
What makes it possible to tackle this case are the strong properties of a viable set. Suppose that $S$ is a viable set for $T$, the set of states of energy at most $\eta$. Then $S$ is also a viable set for $T'$, the set of states of energy at most $\eta - \mu$ for an arbitrary choice of $\mu$. Now, if we apply an AGSP which amplifies the norm of states with energy less than $\eta - \mu$, and decreases the norm of states with energy greater than $\eta$, this is guaranteed to improve the quality of the viable set. This is because by Lemma \ref{l:2}, for each state in $T'$ the viable set $S$ 
contains an approximation to that state that is guaranteed to have no projection onto the orthogonal complement of $T'$ in $T$. In this sense, regarding $S$ as a viable set for $T'$ 
creates a virtual spectral gap $\mu>0$. 

Due to the absence of a constant spectral gap, and our introduction of an ``artificial'' gap of order $1/\log n$, the procedure now runs in quasipolynomial time  $e^{\tilde{O}(\log^3 n)}$. The result is stated in the following theorem. 

\begin{theorem}\label{t:alg2}  Let $H$ be a local Hamiltonian satisfying Assumption (LD), $\eta>0$ the associated energy parameter, $\mu = \Omega(1/\log n)$ and $\delta \geq \poly^{-1}(n)$. With probability at least $1- \frac{1}{n}$ the set $\{ \ket{z_i}\}$ returned by \low$(H,\delta,(\eta,\mu))$ is an orthonormal set of $r$ states each having energy at most $\eps_0 + \eta - \mu + \delta$ with respect to $H$.  The algorithm runs in time  $2^{\poly\log(n)}$.
\end{theorem}

\begin{proof}
The proof follows closely that of Theorem~\ref{t:alg1} with the following simple modifications: Theorem \ref{t:aisefficient} is replaced by Theorem \ref{t:aisgaplessefficient}, there is no need to introduce \mpr' (since the final running time we are obtaining is already $n^{\poly \log n}$ anyways), and finally Theorem~\ref{thm:al-dg} is replaced by Theorem~\ref{thm:al-ld}; as a consequence  any choice of $\xi$ for the final step that is of order 
\begin{equation}\label{eq:algo-xi-ld}
\xi = e^{-\log^{1+\omega(1)} n}
\end{equation}
will suffice to guarantee that trimming induces an error that is negligible compared to $\delta = \poly^{-1}(n)$. 
\end{proof}

We note that we cannot make the stronger conclusion that the $r$ vectors $\ket{z_i}$ returned by the algorithm are low-energy eigenstates; while it does hold that each must have energy at most $\eps_0+ \eta - \mu + \delta$ (since the  closest vectors to $H_{[\eps_0, \eps_0+ \eta - \mu]}$ in $W$ will have this property), in the absence of a spectral gap for $H$ the $\ket{z_i}$ may still be constituted of a mixture of low-energy eigenstates with energy slightly higher than $\eps_0+ \eta - \mu + \delta$.

\section{Discussion}
\label{sec:discussion}

We have introduced a framework for designing algorithms (and proving area laws) by combining procedures for efficiently manipulating viable sets. The scope and efficiency of the resulting algorithms depend upon the efficiency of these procedures. The central limiting component is the efficiency of the underlying AGSP constructions: any substantial improvement of the parameters of our constructions would almost automatically lead to improved area laws, faster algorithms, possibly for scenario that we are currently unable to handle. This naturally leads to a program of determining the 
ultimate limits for these parameters and efficiency bounds, and in particular to the following questions:
\begin{enumerate}
\item What is the best $D-\Delta$ trade-off achievable for an AGSP, depending on assumptions placed on the local Hamiltonian from which it originates? Currently our trade-offs take the form $2^{\log^{3/2 - \mu} D} \Delta < 1$ for arbitrarily small constant $\mu$. Is a better tradeoff achievable, with a larger exponent than $3/2$? Note that currently we only make use of trade-offs of the form $D^c\Delta < 1$ for constant $c$, which is already implied by the above with exponent $1$ instead of $3/2$. Improving the exponent could help make progress towards an area law for 2D systems. For a given trade-off, a related question asks for the smallest value  
of $D$ for which $D\Delta < 1$: this value is important for the efficiency of the algorithm, and also directly enters the parameters obtained for the area law by the bootstrapping argument. Currently, our constructions achieve $ D \simeq \exp(\log^3 d /\gamma)$. 
\item The soft truncation procedure used for the AGSP construction for our algorithms achieves a $\poly(n) $ bond dimension at all cuts. Could that dimension be lowered, perhaps to polylogarithmic at all cuts? 
\item Is it possible to construct an AGSP with a favorable $D-\Delta$ trade-off, not only at one, two, or a constant number of pre-specified cuts, but simultaneously at every (or a constant fraction of) cut?
\item Our trimming procedure for viable sets is not completely satisfactory, and its dependence on the number of cuts as well as on the viable space dimension could potentially be improved. Could the more simple trimming procedure of~\cite{ref:LVV2013-1Dalg} also be applied in this setting?
\item For the case of an MPS approximation to a unique ground state, the parallel  trimming procedure used in  \cite{ref:LVV2013-1Dalg} yields a bound on the trimmed bond dimension that depends inverse-linearly on the desired approximation error, multiplied by the number of bonds trimmed.  It is not impossible that the same procedure would be more effective than proven, with a cost that does not scale with the number of bonds. Such a procedure could yield a nearly-linear time algorithm for the frustration-free case.   
\item What implications can be drawn from our results for the challenging scenario of ground states of local Hamiltonians in higher dimension --- e.g. on 2D lattices? Difficulties such as the efficiency of contracting 2D PEPS networks present significant obstacles to any algorithmic procedure; nevertheless it could be that our bootstrapping arguments could be ported to yield mild area laws in higher dimensions. 
\item The tensor network picture of our algorithm may have an interesting interpretation in terms of the bulk-boundary correspondence in AdS/CFT (see e.g.~\cite{maldacena2003eternal,pastawski2015holographic}). Specifically, the physical qubits would constitute the ``boundary", 
and are acted on directly by the AGSP, while the bulk degrees of freedom are the ones that are subject to the random sampling. 

\end{enumerate}

\paragraph{Acknowledgements}  
We thank Andras Molnar for comments on an earlier draft of this paper, and Christopher T. Chubb for comments and the permission to include the suggestive pictures representing the tensor network structure of the isometry produced by our algorithms. We thank the anonymous referees for valuable comments that greatly helped improve the presentation of the paper.

I. Arad's research was partially performed at the Centre for Quantum Technologies, funded by the Singapore Ministry of Education and the National Research Foundation, also
through the Tier 3 Grant random numbers from quantum processes. 
Z. Landau and U. Vazirani acknowledge support by ARO Grant 
W911NF-12-1-0541, NSF Grant CCF-1410022 and Templeton Foundation Grant 52536
T. Vidick was partially supported by the IQIM, an NSF Physics Frontiers
Center (NFS Grant PHY-1125565) with support of the Gordon and Betty
Moore Foundation (GBMF-12500028).  

\appendix

\section{Constructing an MPO for the cluster expansion}
\label{sec:cluster-mpo}

 In this appendix we provide 
the proof of Theorem~\ref{thm:cluster-approx} and Theorem~\ref{thm:cluster} from Section~\ref{sec:soft-truncation}.

\begin{proof}[Proof of Theorem~\ref{thm:cluster-approx}]
  For an integer $m\geq 1$ we let $\rho_m$ be the summation of
  $f(w)$ over all words $w$ such that there exists $m$ disjoint
  intervals, each of length at least $r$, such that the support of
  $w$ contains each interval but does not contain the two qudits
  that lie immediately to the left and right of the interval (we
  call these two qudits the ``boundary'' of the interval). Using the
  inclusion-exclusion principle one can verify that 
  \begin{equation}
  \label{eq:cluster-1}
    e^{-\beta H} -M_r(H) = -\sum_{m=1}^\infty (-1)^m \rho_m.
  \end{equation}
  We bound the operator norm of each $\rho_m$ individually. Write
  $\rho_m = \sum_{I=\{I_1,\ldots,I_m\}}\rho_I$, where the summation
  is over all $m$-tuples of disjoint intervals $I_1,\ldots,I_m$ of
  length at least $r$, and each $\rho_I$ contains all those $h_w$
  for which the support of $w$ contains each of the intervals $I_i$
  but not its boundary and is arbitrary everywhere else. Very
  roughly, the summation is over at most $n^{2m}/(m!)$ terms. Using
  that the boundaries are excluded, it is not hard to see that
  $\rho_I = e^{-\beta H_{\overline{I}}} \prod_{j=1}^m \eta(I_j)$,
  where $H_{\overline{I}}$ contains all terms in the Hamiltonian
  that do not act on the qudits in the boundary of $I_j$ and
  $\eta(I_j)$ is the sum of all $f(w)$ such that the support of $w$
  is exactly $I_j$. Using $\|e^{-\beta H_{\overline{I}}}\|\leq 1$ we
  can bound
  \begin{align*}
    \norm{\rho_I} &\leq \prod_j \norm{\eta(I_j)} \\
    &\leq \prod_j\Big(\sum_{w:\,supp(w)=I_j} 
      \frac{(-\beta)^{|w|}}{|w|!}\Big)\\
    &= \big(e^{\beta}-1\big)^{\sum_j|I_j|}. 
  \end{align*}
  Combining with~\eqref{eq:cluster-1}, 
  \begin{align*}
    \norm{e^{-\beta H} -M_r(H)} 
      &\leq \sum_{m=1}^\infty\norm{\rho_m}\\
      &\leq \sum_{m=1}^\infty 
        \sum_{I=\{I_1,\ldots,I_m\}}\norm{\rho_I}\\
    &\leq \sum_{m=1}^\infty \frac{n^{2m}}{m!} 
      \big(e^{\beta}-1\big)^{mr}\\
    &=e^{n^2(e^{\beta}-1)^{r}}-1 \,,
  \end{align*}
  where for the third line we used that $\beta$ is such that
  $e^{\beta}-1 <1$.
\end{proof}

\begin{proof}[Proof of Theorem~\ref{thm:cluster}]
  The $r^{\text{th}}$ expansion of $e^{-\beta H}$ is given by 
  \begin{align*}
    M_r(H) \EqDef \sum_{w\in S_r}\frac{(-\beta)^{|w|}}{|w|!} h_w \,,
  \end{align*}
  where $w$ is a word on the alphabet of local Hamiltonian terms
  $\{1, \ldots, n-1\}$, $h_w\EqDef \prod_{i\in w}h_i$, and $S_r$ is
  the set of words in which all connected components have a support
  of size at most $r-1$. Let $\uI=(I_1, I_2, \ldots, I_m)$ be a
  collection of disjoint segments on the line, and $\max(\uI)$
  denote the length of the largest segment in $\uI$. We write $w\in
  \uI$ to mean that the connected components of $w$ matches the
  segments specified by $\uI$. Using this notation, $M_r(H)$ can be
  rewritten as
  \begin{align*}
    M_r(H) = \sum_{\max(\uI)<r} 
      \sum_{w\in \uI} \frac{(-\beta)^{|w|}}{|w|!} h_w \,.
  \end{align*}
  A rather straightforward combinatorial argument shows that 
  for a given $\uI=(I_1, \ldots, I_m)$,
  \begin{align*}
    \sum_{w\in \uI} \frac{(-\beta)^{|w|}}{|w|!} h_w
      = \prod_{j=1}^m \sum_{w\in I_j} 
        \frac{(-\beta)^{|w|}}{|w|!} h_w \,,
  \end{align*}
  where the notation $w\in I_j$ means that the support of the word
  $w$ has a single connected component whose support is $I_j$.
  Therefore, if we define for each segment $I$
  \begin{align}
    \rho_I \EqDef \sum_{w\in I} 
        \frac{(-\beta)^{|w|}}{|w|!} h_w \,,
  \end{align}
  then
  \begin{align}
  \label{eq:rho-expansion}
    M_r(H) = \sum_{\max(\uI)<r}
      \rho_{I_1}\otimes\rho_{I_2}\otimes\cdots \otimes \rho_{I_m}\,.
  \end{align}
  We use \Eq{eq:rho-expansion} as the basis for an efficient MPO
  representation of $M_r(H)$.

  \paragraph{1st step: creating a table of $\rho_I$} The first step
  is a pre-processing step, which can be run performed before the
  start of the algorithm. Its goal is to create a table of MPO
  representations of \emph{all} $\rho_I$ that appear in
  \Eq{eq:rho-expansion}. This can be done in $nd^{O(r)}$ time.
  Indeed, note first that the total number of intervals $I$ to
  consider is at most $nr$. The associated MPO can be computed
  iteratively, starting with $I=\emptyset$ for which
  $\rho_\emptyset=\Id$. Assuming all $\rho_I$ with $|I|<s$ have been
  determined, compute an MPO for $\rho_I$, for any $I$ such that
  $|I|=s$, as follows.  Clearly, 
  \begin{align*}
    \rho_I = e^{-\beta H_I} - \sum_{\uI'} 
      \rho_{I'_1}\otimes\rho_{I'_2}\cdots \otimes\rho_{I'_m},
  \end{align*}
  where the summation runs over all disjoint subsets $\uI'=(I'_1,
  I'_2, \ldots I'_m)$ included in $I$ and with $m\geq 2$. An MPO for
  the first term can be obtained in time $d^{O(s)}$ by direct matrix
  exponentiation. The second term is expressed as the sum of most
  $2^s$ terms, for each of which an MPO was computed in a previous
  iteration. Altogether $\rho_I$ can therefore be computed in time
  $d^{O(s)}$ and stored in memory as an MPO of bond dimension at
  most $d^r$. 
  
  \paragraph{2nd step: creating the MPO of $M_r(H)$} We
  follow the expansion~\Eq{eq:rho-expansion}, using a signaling
  mechanism through which every site tells the site to its {right}
  to which $\rho_I$ it belongs. This ensures that every non-vanishing
  contraction of the virtual indices corresponds to exactly one
  product $\rho_{I_1}\otimes\rho_{I_2}\otimes\cdots$ from
  \Eq{eq:rho-expansion}. 
  
  Virtual bonds are indexed by triples $(\ell,k,\alpha)$. The
  virtual bond across sites $a, a+1$ describes the segment $I$ to
  which $a$ belongs: $\ell\in \{0,1,\ldots, r-1\}$ denotes the width
  of $I$, $k\in \{1,\ldots, r-1\}$ denotes the position of the site
  $a$ within $I$, and $\alpha$ corresponds to the index of the
  virtual bond in the MPO expansion of $\rho_I$.  For example,
  suppose that site $a$ is in third position in the support of
  $\rho_I$, where $|I|=8$. Then it transmits to site $a+1$ the
  indices $\ell=8, k=3$. Site $a+1$ will then transmit to $a+2$ the
  indices $\ell=8, k=4$ and so on. When the last site in $\rho_I$ is
  reached, in our example site $a+5$, it transmits to $a+6$ the
  indices $(k=8,\ell=8)$. Then $a+6$ could either be an empty site,
  transmitting $\ell=k=0$ to the right, or start a new segment $I$
  with any $\ell>0$. See \Fig{fig:transmission} for an illustration
  of this signaling mechanism.
  \begin{figure}
    \begin{center}
      \includegraphics[scale=1]{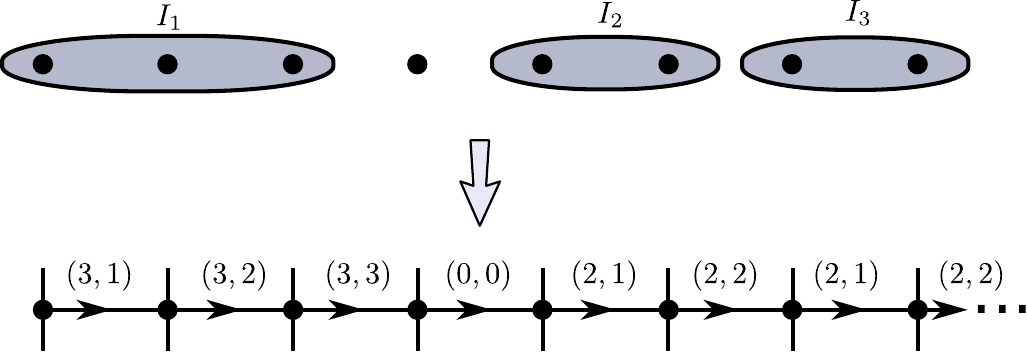}
      \caption{\label{fig:transmission}
      An example of the $(\ell,k)$ indices that give rise to the
      configuration of segments $I_1=(1,2,3); I_2=(5,6); I_3=(7,8)$.}
    \end{center}
  \end{figure}
  
  To write a formal definition of the MPO, let us use
  $[A^{(a)}(I)]^{i,j}_{\alpha_1,\alpha_2}$ to denote the tensor
  associated with $\rho_I$ at site $a\in I$. In order to simplify
  notation, when the site $a$ is the left-most (resp. right-most)
  site in $I$ we use the convention that
  $[A^{(a)}(I)]^{i,j}_{\alpha_1,\alpha_2}$ is non-vanishing only
  when $\alpha_1=1$ (resp. $\alpha_2=1$).  Finally, we denote each
  segment $I$ by $I(\ell, a)$ where $\ell$ is the width of the
  segment and $a$ is its first site. For a non-extremal site $a$,
  the tensor $A^{(a)}$ of $M_r(H)$ is given by
  \begin{align}
   &[A^{(a)}]^{i,j}_{(\ell_1, k_1, \alpha_1), (\ell_2, k_2, \alpha_2)}\EqDef \notag\\
    &  
			\begin{cases}
        \big[A^{(a)}\big(I(\ell_1,a-k_1+1)\big)\big]^{i,j}_{\alpha_1, \alpha_2} 
          \  & \text{for $k_1<\ell_1$ and 
            $\ell_1=\ell_2$, and $k_2=k_1+1$,} \\
        \big[A^{(a)}\big(I(\ell_2,a)\big)\big]^{i,j}_{\alpha_1, \alpha_2} 
          \  & \text{for $k_1=\ell_1$ and $0<\ell_2\le n-a+1$
             and $k_2=1$,} \\
        \delta_{i,j}
          \  & \text{for $k_1=\ell_1$ and $\ell_2=k_2=0$ and
          $\alpha_1=\alpha_2=1$,}\\
        0
          \  & \text{otherwise.}
      \end{cases} \label{eq:Atensor}
  \end{align}
  The first case corresponds to a site $a$ in the interior of the
  segment $I = I(\ell_1,a-k_1+1)$.  The second case corresponds to
  an $a$ that is the first site of a new segment $I=I(\ell_2, a)$.
  Note that the condition $\ell_2\le n-a+1$ guarantees that this
  segment does not exceed the right side of the chain. Finally, the
  third case corresponds to an empty site $a$. 
  
  To complete the definition it remains to specify $A^{(1)}$ and
  $A^{(n)}$. Just as the tensors for $\rho_I$, we keep both left and
  right indices but make them non-zero only when $\ell=k=0$ and
  $\alpha=1$. Then $A^{(1)}$ is defined as $A^{(a)}$ with the
  additional requirement that it is non-vanishing only when
  $\ell_1=k_1=0$ and $\alpha_1=1$. The tensor $A^{(n)}$ is defined
  directly by \eqref{eq:Atensor}. In that case, for every
  $(\ell_1,p_1,\alpha_1)$ there is at most one triple
  $(\ell_2,p_2,\alpha_1)$ for which $A^{(n)}$ is non-vanishing, and
  so without loss of generality we can map it to $\ell_2=k_2=0$ and
  $\alpha_2=1$.
  
  To finish the proof note that the vritual bond dimension is 
  bounded by $r(r-1)d^r<r^2d^r$, and therefore the second step can
  be done in time $nd^{\bigO{r}}$ since it only involves local
  assignments.
\end{proof}

\bibliography{algo1D}

\end{document}